\documentclass{IEEEtran}

\usepackage[cmex10]{amsmath}
\usepackage{amsfonts,amssymb,amsbsy}
\usepackage{graphicx,stfloats,comment}

\newtheorem{theorem}{Theorem}

\newtheorem{example}{Example}

\newtheorem{proposition}{Proposition}
\newtheorem{remark}{Remark}
\newenvironment{proof}[1][Proof]{{\em #1:} }{\ \rule{0.5em}{0.5em}}

\newcommand{\set}[1]{\mathcal{#1}} 
\newcommand{\supp}{{\rm supp}}
\newcommand{\ul}[1]{\underline{#1}}
\newcommand{\E}[1]{{\rm E}\left[{#1}\right]}

\begin{document}

\author{
\IEEEauthorblockN{Gerhard Kramer}
\IEEEauthorblockA{
}
\thanks{Date of current version \today.
Accepted for publication in the IEEE Transactions on Information Theory, November 18, 2013.
G.~Kramer was supported in part by an Alexander von Humboldt Professorship through
the German Federal Ministry of Education and Research and in part by NSF under Grant CCF-09-05235.
The paper was presented in part at the
2012 First Workshop on Information Theory and Coding for Cooperative Networks,
the 2012 First Munich Workshop on Bidirectional Communication
and Directed Information, Munich, Germany, and the 2012
IEEE Information Theory Workshop.

The author is with the Institute for Communications Engineering,
Technische Universit\"at M\"unchen, Munich D-80333, Germany
(email: gerhard.kramer@tum.de).

\bigskip

}
}

\title{Information Networks With In-Block Memory}

\maketitle

\begin{abstract}
A class of channels is introduced for which there is memory inside blocks
of a specified length and no memory across the blocks. The multi-user model
is called an information network with in-block memory (NiBM).
It is shown that block-fading channels, channels with state known causally
at the encoder, and relay networks with delays are NiBMs.
A cut-set bound is developed for NiBMs that unifies, strengthens, and generalizes
existing cut bounds for discrete memoryless networks. The bound gives
new finite-letter capacity expressions for several classes of networks including point-to-point
channels, and certain multiaccess, broadcast, and relay channels.
Cardinality bounds on the random coding alphabets are developed that improve on
existing bounds for channels with action-dependent state available causally at the
encoder and for relays without delay.
Finally, quantize-forward network coding is shown to achieve rates within an additive
gap of the new cut-set bound for linear, additive, Gaussian noise channels, symmetric
power constraints, and a multicast session.
\end{abstract}


\section{Introduction}
\label{sec:intro}
Communication channels often have memory, e.g., due to bandwidth limitations
and dispersion. The memory is often modeled as being finite and of a sliding-window
type, e.g., a convolution. However, in a network environment with bursty traffic and
interference one often schedules users to dedicated time-frequency slots and with
time-frequency offsets between successive slots. A pragmatic approach is then to
model the channel as having memory inside a block and as being memoryless
across blocks. We say that such channels have {\em in-block memory} or iBM.

This paper studies {\em networks} with iBM (NiBMs) where two central themes are
{\em memory} and {\em feedback}.
Several classes of channels fall into the NiBM framework, including block-fading
channels~\cite{OzarowShamaiWyner94}, channels with state known causally at the
encoder~\cite{Shannon58}, and relay networks with delays~\cite{ElGamal:07}.
In fact, the original motivation for this work was to show that the theory for relay
networks with delays can be derived from theory for discrete memoryless networks
(DMNs).
We only later realized that NiBMs include block fading channels and channels with state known causally at the encoders.

This document is organized as follows. Section~\ref{sec:model} presents the NiBM model.
Section~\ref{sec:prelim} defines the capacity region of a NiBM and introduces notation.
Section~\ref{sec:cut-set-bound} states our main technical result:\ a cut-set bound on
reliable communication rates. Sections~\ref{sec:ptp} and~\ref{sec:multiuser} apply
the bound to point-to-point and multiuser channels, and they show that NiBMs let us
unify, strengthen, and generalize existing theory for several classes of networks.
For example, we derive new capacity theorems and new cardinality bounds on
random variables. Section~\ref{sec:relay-networks} extends the approaches
to relay networks. Several proofs are developed in the Appendices.

\section{Model}
\label{sec:model}
The general DMN model was studied in~\cite{vanderMeulen:68} and a bounding
tool for a class of DMNs called relay networks was developed in~\cite{GamalA:81}
(see also~\cite{Cover06}). We use terminology and notation from~\cite{Kramer03}.
Recall that a DMN with $K$ nodes has each node $k$, $k=1,2,\ldots,K$, dealing
with four types of random variables.
\begin{itemize}
\item  {\em Messages} $W_{km}$, $m=1,2,\ldots,M_k$, that have entropy $H(W_{km})=B_{km}$ bits
          where $M_k$ is the number of messages at node $k$.
          The rate of message $W_{km}$ is
          thus $R_{km}=B_{km}/n$ bits per channel use. The $\{W_{km}\}$ are mutually statistically independent
          for all $m$ and $k$.
\item {\em Channel inputs} $X_{k,i}$, $i=1,2,\ldots,n$, with alphabet $\set{X}_k$. We interpret $i$ as a time index
          but it could alternatively represent frequency or space, for example.
\item {\em Channel outputs} $Y_{k,i}$, $i=1,2,\ldots,n$, with alphabet $\set{Y}_k$.
\item {\em Message estimates} $\hat{W}_{\ell m}^{(k)}$, $\ell m \in\set{D}(k)$, where
          $\set{D}(k)$ is a {\em decoding index set} whose
          elements are selected pairs $\ell m$, $\ell\ne k$, of message indices from other nodes.
\end{itemize}

Let $\set{K}=\{1,2,\ldots,K\}$ be the set of nodes;
let $\set{E}(k)=\{k1,k2,\ldots,kM_k\}$ be the {\em encoding index set} of node $k$;
let $Y_{k}^i=Y_{k,1} Y_{k,2} \ldots Y_{k,i}$;
let $r(x,y)$ be the remainder when $x$ is divided by $y$.
For a set $\set{S} \subseteq \set{K}$ we write
$\set{E}(\set{S})=\cup_{k\in\set{S}} \set{E}(k)$ and $X_{\set{S},i} = \{ X_{k,i} : k \in \set{S} \}$.
For a set $\tilde{\set{S}}$ of integer pairs $km$ we write
$W_{\tilde{\set{S}}}=\{W_{km}: km\in\tilde{\set{S}}\}$.
The relationships between the random variables are as follows. 
\begin{itemize}
\item  Without feedback, node $k$ chooses $X_{k,i}$ as a function of $W_{\set{E}(k)}$ only.
          The $X_k^n(W_{\set{E}(k)})$ are called {\em codewords}.

\item  With feedback, node $k$ chooses functions ${\mathbf a}_{k,i}$, $i=1,2,\ldots,n$, such that
           \begin{align} \label{eq:channel-input}
              X_{k,i} = {\mathbf a}_{k,i}(W_{\set{E}(k)},Y_k^{i-1}).
           \end{align}
           We call ${\mathbf a}_k^n(W_{{\set E}(k)},\cdot)$ a {\em code function} 
           or an {\em adaptive codeword} since it replaces the notion of a codeword.
           For a finite alphabet $\set{Y}_k$ one may
           interpret ${\mathbf a}_k^n(W_{{\set E}(k)},\cdot)$ as a {\em code tree}
           (see \cite[Sec.~15]{Shannon60}, \cite[Sec.~5]{vanderMeulen:68}, and \cite[Ch.~9]{Blahut87}).
           We write ${\mathbf a}_k^n(W_{{\set E}(k)},\cdot)$ as  ${\mathbf A}_k^n(W_{{\set E}(k)},\cdot)$
           when we wish to emphasize that ${\mathbf A}_k^n$ is a random variable.
           The alphabet of ${\mathbf A}_k^n(W_{{\set E}(k)},\cdot)$ is written as $\set{A}_k^n$ and
           for finite $\set{X}_k$ and $\set{Y}_k$ we have the cardinality
\begin{align} \label{eq:card-A}
   \left|\set{A}_k^n\right| = \prod_{i=1}^n |\set{X}_k|^{| \set{Y}_k |^{i-1}} .
\end{align} 
           For example, if all alphabets are binary and $n=3$ then there are
           2 choices for $X_{k,1}$, 2 choices for $X_{k,2}$ for each of the 2 
           possible $Y_{k,1}$, and 2 choices for $X_{k,3}$ for each of the 4
           possible $Y_{k,1}Y_{k,2}$. The result is $2^1\cdot2^2\cdot2^4=128$
           possible code trees ${\mathbf A}_k^3$.

\item A DMN channel is memoryless and time-invariant in the sense that
         at time $i$ node $k$ receives
\begin{align} \label{eq:Yfd-DMN}
   Y_{k,i} = f_{k}\left( X_{\set{K},i}, Z_i \right)   
\end{align}
        for some functions $f_{k}(\cdot)$, $k=1,2,\ldots,K$, where the $Z_i$, $i=1,2,\ldots,n$, are
        statistically independent realizations of a noise random variable $Z$ with alphabet $\set{Z}$.
Instead, a NiBM may have in-block memory (iBM)
with block length $L$ (or memory $L-1$) in the sense that
\begin{align} \label{eq:Yfd}
   Y_{k,i} = f_{k,t(i)+1}\left( X_{\set{K},i-t(i)}, \ldots, X_{\set{K},i}, Z_{\lceil i/L \rceil} \right)   
\end{align}
for some functions $f_{k,i}(\cdot)$, $k=1,2,\ldots,K$, $i=1,2,\ldots,L$, where $t(i)=r(i-1,L)$,
and where the $Z_i$, $i=1,2,\ldots,\lceil n/L \rceil$, are
statistically independent realizations of a noise random variable $Z$ with alphabet $\set{Z}$.
The noise $Z$ could be a vector random variable.

\item  Node $k$ puts out the message decisions
           \begin{align}
              \hat{W}_{\set{D}(k)}^{(k)} = d_k(W_{\set{E}(k)},Y_k^n)
           \end{align}
           for some decoding function $d_k$.
           
\end{itemize}

\begin{example}
Consider a two-way channel with iBM and block length $L=2$. The channel puts out
\begin{itemize}
\item $Y_{k,1} = f_{k,1}(X_{1,1},X_{2,1},Z_1)$
\item $Y_{k,2} = f_{k,2}(X_{1,1},X_{1,2},X_{2,1},X_{2,2},Z_1)$
\item $Y_{k,3} = f_{k,1}(X_{1,3},X_{2,3},Z_2)$
\item $Y_{k,4} = f_{k,2}(X_{1,3},X_{1,4},X_{2,3},X_{2,4},Z_2)$
\end{itemize}
for $k=1,2$ and $n=4$. A functional dependence graph (FDG) for this channel is shown in
Fig.~\ref{fig:twc-block-memory-2} where the nodes $W_1$, $W_2$, $Z_1$, $Z_2$
with hollow circles represent mutually statistically independent random
variables~\cite{Kramer03,Kramer07}.
\end{example}

\begin{figure*}[t!]
  \centerline{\includegraphics[scale=0.5]{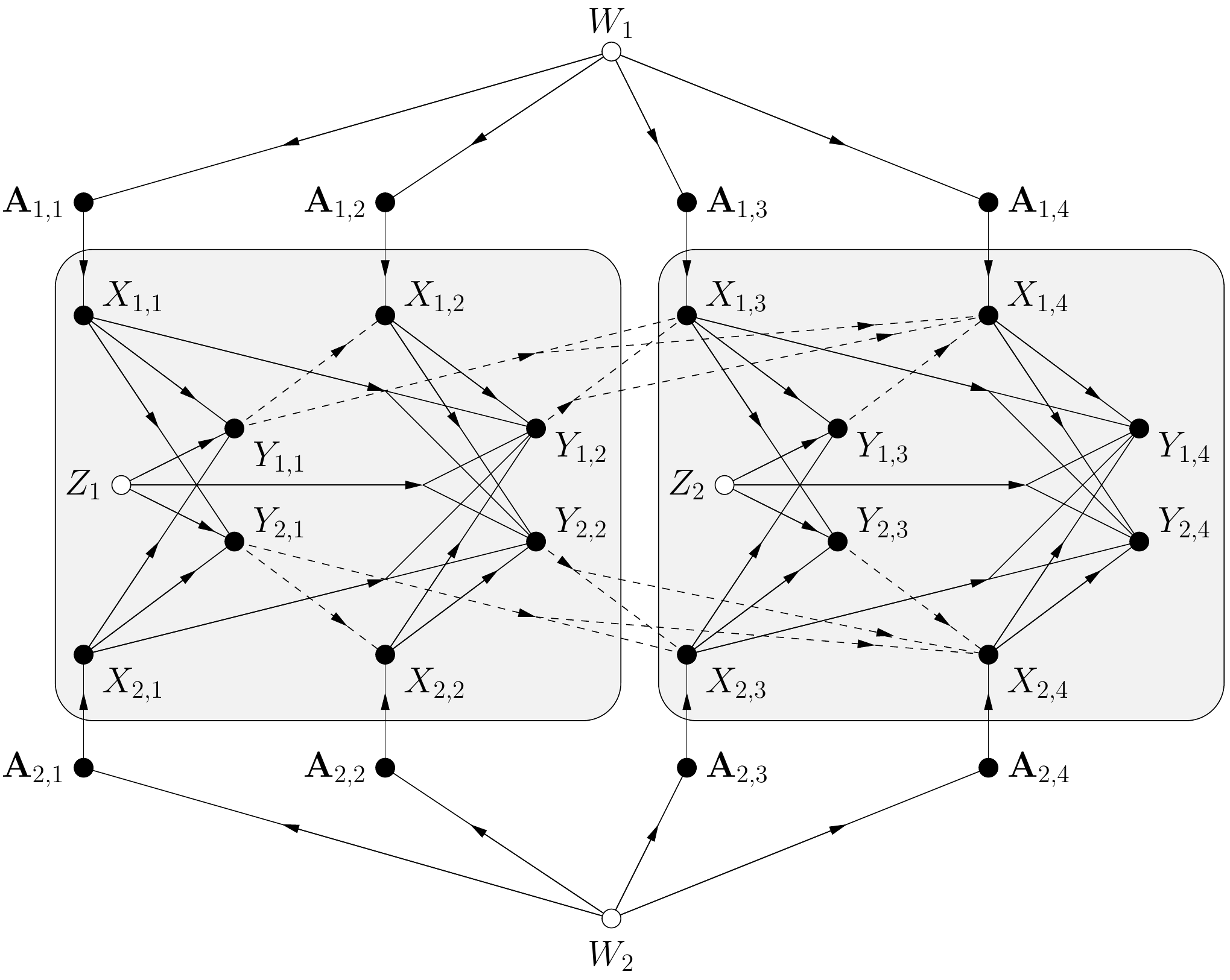}}
  \caption{FDG for a two-way channel with iBM and block length $L=2$ for $n=4$ channel uses.
  The message estimates $\hat{W}_1$ and $\hat{W}_2$ are not shown.
  The two blocks of channel inputs and outputs are shaded and the functional dependence
  due to the received symbols is drawn with dashed lines. The code functions
  ${\mathbf A}_1^n$ and ${\mathbf A}_2^n$ are statistically independent.}
  \label{fig:twc-block-memory-2}
\end{figure*}

\begin{remark} \label{remark:no-feedback}
   Without feedback, the NiBM becomes a DMN if we view blocks of $L$ letters
   as a single letter, i.e., we have a DMN with {\em vector} inputs and outputs.
\end{remark}

\begin{remark}
   For time-varying channels the input and output alphabets of node $k$
   may be different for different times $i$. In this case, we write the alphabets as $\set{X}_{k,i}$ and
   $\set{Y}_{k,i}$, $i=1,2,\ldots,L$. The notation $\set{X}_k^L$ means
   $\set{X}_{k,1} \times \set{X}_{k,2} \times \ldots \times \set{X}_{k,L}$.
\end{remark}

\section{Preliminaries}
\label{sec:prelim}
%
\subsection{Capacity}
\label{subsec:capacity}
The {\em capacity region} $\set{C}$ of a NiBM is the closure of the set of rate-tuples
$(R_{km}: 1\le k \le K, 1\le m \le M_k)$ such that for any positive $\epsilon$ there is
an $n$ and code functions and decoders for which the error probability
\begin{align}
P_e = \Pr \left[ \bigcup_{k} \bigcup_{\ell m \in \set{D}(k)}
           \{ \hat{W}_{\ell m}^{(k)} \ne W_{\ell m} \} \right]
\end{align}
is at most $\epsilon$.

\subsection{Causal Conditioning and Directed Information}
\label{subsec:cc}
We use notation from~\cite{Kramer03} for causal conditioning and directed information
(see also~\cite{Massey90,Kramer98,Tatikonda00}).
The probability of $x^L$ causally conditioned on $y^L$ and conditioned on $a$ is defined as
\begin{align}
  & P(x^L \| y^L) = \prod_{i=1}^L P(x_i | x^{i-1}, y^i) \\
  & P(x^L \| y^L | a) = \prod_{i=1}^L P(x_i | x^{i-1}, y^i, a).
\end{align}
As done here, we will drop subscripts on probability distributions if the
argument is the lowercase version of the random variable.
Causally-conditioned entropy is defined as
\begin{align}
  & H(X^L \| Y^L) = \sum_{i=1}^L H(X_i | X^{i-1} Y^i) \\
  & H(X^L \| Y^L | A) = \sum_{i=1}^L H(X_i | X^{i-1} Y^i A)
\end{align}
where the notation $X^{i-1} Y^i$ refers to the concatenation of $X^{i-1}$ and $Y^i$.
Directed information is written as
\begin{align}
  & I(X^L \rightarrow Y^L) = H(Y^L) - H(Y^L \| X^L) \\
  & I(X^L \rightarrow Y^L \| Z^L) = H(Y^L \| Z^L) - H(Y^L \| X^L,Z^L) \label{eq:ccdi} \\
  & I(X^L \rightarrow Y^L \| Z^L | A) \nonumber \\
  & = H(Y^L \| Z^L | A) - H(Y^L \| X^L,Z^L | A). \label{eq:ccdi2}
\end{align}
The commas in \eqref{eq:ccdi} and \eqref{eq:ccdi2} emphasize
that the pair $X^L,Z^L$ should here be considered as a
length-$L$ sequence of pairs $(X_1,Z_1), (X_2,Z_2),\ldots, (X_L,Z_L)$.
As another example of such notation, we write the directed information flowing from
$X_1^L$, $X_2^L$ to $Y^L$ when causally conditioned on $Z_1^L$, $Z_2^L$ as
\begin{align}
  & I(X_1^L,X_2^L \rightarrow Y^L \| Z_1^L, Z_2^L) \nonumber \\
  & = H(Y^L \| Z_1^L, Z_2^L) - H(Y^L \| X_1^L, X_2^L, Z_1^L, Z_2^L).
\end{align}

\subsection{Further Notation}
\label{subsec:further-notation}
The functional dependence \eqref{eq:channel-input} implies that
$P(x_{k,i} | {\mathbf a}_{k}^i, y_{k}^{i-1})$ takes on the value 1 only for that
letter $x_{k,i}$ satisfying \eqref{eq:channel-input}, and is 0 otherwise.
To emphasize such dependence, we write
$1(x_{k,i} | {\mathbf a}_{k}^i, y_{k}^{i-1})$ in place of $P(x_{k,i} | {\mathbf a}_{k}^i, y_{k}^{i-1})$,
and similarly
$1(x_{k}^L \| {\mathbf a}_{k}^L, 0y_{k}^{L-1})$ in place of $P(x_{k}^L \| {\mathbf a}_{k}^L, 0y_{k}^{L-1})$.
The expression $0y_{k}^{L-1}$ denotes the concatenation of $0$ and $y_{k}^{L-1}$.

It will be convenient to split symbol strings into blocks of length $L$. We use the notation
\begin{align*}
   {\mathbf a}_{k,i}^L={\mathbf a}_{k,i(L-1)+1} \, {\mathbf a}_{k,i(L-1)+2} \,\ldots\, {\mathbf a}_{k,i(L-1)+L} \\
   x_{k,i}^L=x_{k,i(L-1)+1} x_{k,i(L-1)+2} \,\ldots\, x_{k,i(L-1)+L} \\
   y_{k,i}^L=y_{k,i(L-1)+1} \, y_{k,i(L-1)+2} \,\ldots\, y_{k,i(L-1)+L} .
\end{align*}
We write $\supp(P_X)$ for the support set of $P_X(\cdot)$.
We write the binary entropy function as $H_2(\cdot)$ and differential entropy as $h(\cdot)$.
Logarithms are taken to the base 2.

\subsection{Channel Distribution}
\label{subsec:channel-distribution}
We have defined the channel using the {\em function} \eqref{eq:Yfd}. It will be convenient to alternatively define
the channel by a {\em probability distribution}. Consider  $P({\mathbf a}_{\set{K}}^n,x_{\set{K}}^n,y_{\set{K}}^n)$
that factors as
\begin{align}
& \left[ \prod_{k=1}^K  P({\mathbf a}_k^n) 1(x_k^n \| {\mathbf a}_k^n,0y_k^{n-1}) \right]
P(y_{\set{K}}^n \| x_{\set{K}}^n ) .
\label{eq:joint-pdf}
\end{align}
The $P(y_{\set{K}}^n \| x_{\set{K}}^n )$ further factors into $m=\lceil n/L \rceil$ blocks as
\begin{align}
\left[ \prod_{i=1}^{m-1} P_{Y_{\set{K}}^L \| X_{\set{K}}^L}(y_{\set{K},i}^L \| x_{\set{K},i}^L ) \right]
P_{Y_{\set{K}}^{L'} \| X_{\set{K}}^{L'}}(y_{\set{K},m}^{L'} \| x_{\set{K},m}^{L'} )
\label{eq:channel-pdf}
\end{align}
where the last block has length $L'=n-(m-1)L$.
We focus on $n=mL$ so that $L'=L$ and all blocks have length $L$.
\begin{remark} \label{remark:channel}
The expressions \eqref{eq:joint-pdf}-\eqref{eq:channel-pdf} let us define the channel
by using the block-invariant distribution $P(y_{\set{K}}^L \| x_{\set{K}}^L )$ rather than
by using $Z$ and the functions in \eqref{eq:Yfd}. We further have
\begin{align} \label{eq:remark-channel}
   P(y_{\set{K}}^L | {\mathbf a}_{\set{K}}^L ) = P(y_{\set{K}}^L \| x_{\set{K}}^L ).
\end{align}
Thus, we may view the channel as being defined by the functional relations \eqref{eq:Yfd},
by $P(y_{\set{K}}^L \| x_{\set{K}}^L )$, or by $P(y_{\set{K}}^L | {\mathbf a}_{\set{K}}^L )$.
\end{remark}

\subsection{Linear Channels}
\label{subsec:linear-channels}
We consider several examples where the the channel alphabets are the field $\mathbb{F}$.
We write the channel inputs and outputs as vectors $\ul{X}_k=[X_{k,1} \ldots X_{k,L}]^T$
and $\ul{Y}_k=[Y_{k,1} \ldots Y_{k,L}]^T$, respectively. For instance, a {\em scalar}, {\em linear}, and
{\em additive-noise} channel has
\begin{align} \label{eq:linear-channel}
 \ul{Y}_k = \left[ \sum_{j\ne k} {\mathbf G}_{kj} \ul{X}_j \right] + \ul{Z}_k
\end{align}
where the ${\mathbf G}_{kj}$ are $L\times L$ lower-triangular matrices and
$\ul{Z}_k=[Z_{k,1} \ldots Z_{k,L}]^T$, $k\in\set{K}$. The noise $Z_{\set{K}}^L$ is independent of
${\mathbf A}_{\set{K}}^L$.
We write the covariance matrix of a random vector $\ul{X}$ as
${\mathbf Q}_{\ul{X}}$ and its determinant as $|{\mathbf Q}_{\ul{X}}|$.

\section{Cut-Set Bound}
\label{sec:cut-set-bound}
We develop a cut-set bound for NiBMs that generalizes the classic cut-set bound for DMNs.
Consider a set $\set{S}$ of nodes and let $\set{S}^c$ be the complement of
$\set{S}$ in $\set{K}=\{1,2,\ldots,K\}$. We say that $(\set{S},\set{S}^c)$ is a {\em cut separating} a
message $W_{km}$ and its estimate $\hat{W}_{km}^{(\ell)}$ if $k\in\set{S}$ and $\ell\in\set{S}^c$.
Let $\set{M}(\set{S})$ be the set of indexes
(which are integer pairs $km$) of those messages separated from one of their estimates by the
cut $(\set{S},\set{S}^c)$, and let $R_{\set{M}(\set{S})}$ be the
sum of the rates of these messages.

There is a subtlety in that the NiBM can have high mutual information at the start of each
block and low mutual information elsewhere. For example, consider a
point-to-point channel \eqref{eq:linear-channel} where $\mathbb{F}$ is the Galois field
of size two, $K=2$, $L=2$, the channel matrix is
\begin{align*}
 {\mathbf G}_{21} = \begin{bmatrix} 1 & 0 \\ 0 & 0 \end{bmatrix}
\end{align*}
and $\ul{Z}_2=[0\; 0]^T$. We find that using the channel once gives larger mutual
information per letter  than using the channel twice or more. But this fact is not very
interesting because we wish to transmit information
reliably and can (usually) accomplish this only by using the channel often. To avoid such
formal details, we will require that $n=mL$ for a positive integer $m$. Alternatively, we could
require that $n$ be much larger than $L$. We have the following
result that we prove in Appendix~\ref{app:proof}.
\begin{theorem} \label{thm:cut-set-bound}
The capacity region $\set{C}$ of a NiBM with block length $L$ that is used a multiple of $L$ times
satisfies
\begin{align} \label{eq:cut-set-bound}
\set{C} \subseteq \bigcup_{P_{\mathbf A_{\set{K}}^L}} \bigcap_{\set{S} \subset \set{K}}
\set{R}(P_{{\mathbf A}_{\set{K}}^L},\set{S})
\end{align}
where $\set{R}(P_{{\mathbf A}_{\set{K}}^L},\set{S})$ is the set of non-negative rate-tuples satisfying
\begin{align} \label{eq:rate-bound}
R_{\set{M}(\set{S})} \le I( {\mathbf A}_{\set{S}}^L ; Y_{\set{S}^c}^L | {\mathbf A}_{\set{S}^c}^L )/L.
\end{align}
The joint probability distribution $P({\mathbf a}_{\set{K}}^L,x_{\set{K}}^L,y_{\set{K}}^L)$ factors as
\begin{align} \label{eq:cut-set-pdf}
P({\mathbf a}_{\set{K}}^L) \left[ \prod_{k=1}^K  1(x_k^L \| {\mathbf a}_{k}^L,0y_{k}^{L-1}) \right] P(y_{\set{K}}^L \| x_{\set{K}}^L ).
\end{align}
\end{theorem}
\begin{remark}
The code functions in Theorem~\ref{thm:cut-set-bound} are statistically {\em dependent}.
This is different than in Sec.~\ref{sec:intro} where the code functions are independent
(see Fig.~\ref{fig:twc-block-memory-2} and \eqref{eq:joint-pdf}). Similarly, Shannon's outer bound for the two-way channel
\cite[Eq.~(36)]{Shannon60} and the classic cut-set bound for DMNs~\cite{Kramer03}, \cite[Ch.~10]{Kramer07}, \cite[p.~477]{ElGamal-Kim-11} have statistically {\em dependent} inputs (see Sec.~\ref{subsec:dmn}).
\end{remark}
\begin{remark}
The $1(x_k^L \| {\mathbf a}_{k}^L,0y_{k}^{L-1})$, $k=1,2,\ldots,K$, are fixed functions and
$P(y_{\set{K}}^L \| x_{\set{K}}^L )$ is fixed by the channel.
\end{remark}
\begin{remark} \label{remark:concavity}
Remark~\ref{remark:channel} states that we may view the channel as being
$P(y_{\set{K}}^L | {\mathbf a}_{\set{K}}^L )$. This insight is useful for deriving achievable
rates and for computing the cut-set bound (see~\cite[Sec.~5]{vanderMeulen:68}).
For instance, $I( {\mathbf A}_{\set{S}}^L ; Y_{\set{S}^c}^L | {\mathbf A}_{\set{S}^c}^L )$
is concave in $P_{{\mathbf A}_{\set{K}}}^L$.
This result follows by the concavity of $I(A;B|C=c)$ in $P_{A|C=c}$ when $P_{B|AC=c}$ is held fixed, and because
$P(y_{\set{K}}^L | {\mathbf a}_{\set{K}}^L )$ is fixed.
\end{remark}
\begin{remark} \label{remark:auxRV}
The ${\mathbf A}_k^L$ are {\em not} ``auxiliary" random variables, i.e., they are explicit components of the
communication problem  just like the channel inputs $X_k^L$. Moreover, the cardinalities $|\set{A}_k^L|$
are bounded by the channel alphabets (see~\eqref{eq:card-A}).
\end{remark}
\begin{remark}
Average per-letter cost constraints can be dealt with in the usual way (see Remark~\ref {rmk:cost-constraint} below).
More precisely, if we have $J$ cost functions $s_j(\cdot)$ and constraints
\begin{align} \label{eq:cost-constraints}
   \frac{1}{n} \sum_{i=1}^n \E{s_j\left( X_{\set{K},i}, Y_{\set{K},i} \right)} \le S_j, \quad j=1,2,\ldots,J
\end{align}
then one may add the requirement that the union in \eqref{eq:cut-set-bound} is over
distributions \eqref{eq:cut-set-pdf} that satisfy 
\begin{align} \label{eq:cost-constraints2}
   \frac{1}{L} \sum_{i=1}^L \E{s_j\left( X_{\set{K},i}, Y_{\set{K},i} \right)} \le S_j, \quad j=1,2,\ldots,J .
\end{align}
One may treat average per-block cost constraints similarly.
\end{remark}
\begin{remark}
The bound in~\cite[Thm.~4]{ElGamal:07} is almost the same as \eqref{eq:rate-bound} for
relay networks with delays. We discuss these models in more detail
in Remark~\ref{remark:bound} and Sec.~\ref{subsec:relay-networks-delays} below.
\end{remark}
\begin{remark}
Theorem~\ref{thm:cut-set-bound} improves the bounds in~\cite[Thm.~2]{Baik11}
and~\cite[Thm.~1]{FongIT:12} for causal relay networks and generalized networks.
 We discuss these results in Sec.~\ref{subsec:generalized}.
\end{remark}

\subsection{Weakened Bounds}
\label{subsec:weak}
The bound \eqref{eq:rate-bound} may be weakened as follows:
\begin{align}
& I( {\mathbf A}_{\set{S}}^L ; Y_{\set{S}^c}^L | {\mathbf A}_{\set{S}^c}^L ) \nonumber \\
& \overset{(a)}{=} \sum_{i=1}^L H( Y_{\set{S}^c,i} | Y_{\set{S}^c}^{i-1} {\mathbf A}_{\set{S}^c}^L ) - H( Y_{\set{S}^c,i} | Y_{\set{S}^c}^{i-1} {\mathbf A}_{\set{K}}^i ) \nonumber \\
& \le \sum_{i=1}^L H( Y_{\set{S}^c,i} | Y_{\set{S}^c}^{i-1} {\mathbf A}_{\set{S}^c}^i ) - H( Y_{\set{S}^c,i} | Y_{\set{S}^c}^{i-1} {\mathbf A}_{\set{K}}^i ) \nonumber \\
& = I( {\mathbf A}_{\set{S}}^L \rightarrow Y_{\set{S}^c}^L \| {\mathbf A}_{\set{S}^c}^L )
\label{eq:weak1}
\end{align}
where $(a)$ follows by the chain rule for entropy and because
\begin{align} \label{eq:Markov-chain-1}
   \left( {\mathbf A}_{\set{K},i+1}\ldots{\mathbf A}_{\set{K},L} \right)
   - {\mathbf A}_{\set{K}}^iY_{\set{S}^c}^{i-1} -Y_{\set{S}^c,i}
\end{align}
forms a Markov chain.
The bound \eqref{eq:weak1} is further weakened by replacing code functions with channel
inputs and outputs: 
\begin{align}
& I( {\mathbf A}_{\set{S}}^L \rightarrow Y_{\set{S}^c}^L \| {\mathbf A}_{\set{S}^c}^L ) \nonumber \\
& \overset{(a)}{\le} \sum_{i=1}^L H( Y_{\set{S}^c,i} | Y_{\set{S}^c}^{i-1} X_{\set{S}^c}^i )
   - H( Y_{\set{S}^c,i} | Y_{\set{K}}^{i-1} X_{\set{K}}^i {\mathbf A}_{\set{K}}^i ) \nonumber \\
& \overset{(b)}{=}  I( X_{\set{S}}^L, 0Y_{\set{S}}^{L-1} \rightarrow Y_{\set{S}^c}^L \| X_{\set{S}^c}^L )
\label{eq:weak}
\end{align}
where $(a)$ follows because $Y_k^{i-1} {\mathbf A}_k^i$ defines $X_k^i$
and because conditioning cannot increase entropy. Step $(b)$ follows because
\eqref{eq:cut-set-pdf} ensures that the chain
${\mathbf A}_{\set{K}}^L-Y_{\set{K}}^{i-1} X_{\set{K}}^i-Y_{\set{K},i}$ is Markov.
\begin{remark}
The FDG of a NiBM has statistically independent code functions, see Fig.~\ref{fig:twc-block-memory-2}.
We thus have
\begin{align} \label{eq:index-identity}
H( Y_{\set{S}^c,i} | Y_{\set{S}^c}^{i-1} X_{\set{S}^c}^i {\mathbf A}_{\set{S}^c}^L )
=H( Y_{\set{S}^c,i} | Y_{\set{S}^c}^{i-1} X_{\set{S}^c}^i).
\end{align}
However, the identity \eqref{eq:index-identity} may not be valid when considering dependent
code functions such as in Theorem~\ref{thm:cut-set-bound}.
\end{remark}
\begin{remark}
The cut-set bound with the normalized \eqref{eq:weak} in place of the right-hand side of
\eqref{eq:rate-bound} was derived in \cite[Thm.~1]{Baik11} for causal relay networks 
and in~\cite[Thm.~1]{FongIT:12} for generalized networks. 
The authors of~\cite{Baik11,FongIT:12} restrict attention to multiple unicast sessions as in~\cite[Sec.~15.10]{Cover06}.
Theorem~\ref{thm:cut-set-bound} improves these bounds and extends them to multiple multicast sessions.
We discuss these bounds in more detail in Sec.~\ref{subsec:generalized}.
\end{remark}
\begin{example}
Consider {\em additive} noise channels with
\begin{align} \label{eq:additive-noise}
   Y_{k,i} & = f_{k,i}(X_{\set{K}}^i) + Z_{k,i}
\end{align}
for $i=1,2,\ldots,L$, $k=1,2,\ldots,K$, where $Y_{k,i}$, $Z_{k,i}$, and $f_{k,i}(X_{\set{K}}^i)$
take on values in the field $\mathbb{F}$. The noise variables $Z_{\set{K}}^L$ are independent of
${\mathbf A}_{\set{K}}^L$. For finite fields, the bound~\eqref{eq:weak} is
\begin{align}
I( {\mathbf A}_{\set{S}}^L ; Y_{\set{S}^c}^L | {\mathbf A}_{\set{S}^c}^L )
& \le I( X_{\set{S}}^L, 0Y_{\set{S}}^{L-1} \rightarrow Y_{\set{S}^c}^L \| X_{\set{S}^c}^L ) \nonumber \\
& = H(Y_{\set{S}^c}^L \| X_{\set{S}^c}^L ) - H(Z_{\set{S}^c}^L \| 0Z_{\set{S}}^{L-1}).
\label{eq:weak2}
\end{align}
Since $H(Z_{\set{S}^c}^L \| 0Z_{\set{S}}^{L-1})$ is fixed by the channel, the cut-set bound with the normalized
\eqref{eq:weak2} in place of the right-hand side of \eqref{eq:rate-bound} is a maximum
(conditional) entropy problem.
\end{example}

\begin{example}
A special case of~\eqref{eq:additive-noise} is a {\em deterministic} NiBM for which
the noise is a constant and
\begin{align}
I( {\mathbf A}_{\set{S}}^L ; Y_{\set{S}^c}^L | {\mathbf A}_{\set{S}^c}^L )
& \le H(Y_{\set{S}^c}^L \| X_{\set{S}^c}^L ).
\label{eq:weak3}
\end{align}
\end{example}

\subsection{DMNs}
\label{subsec:dmn}
For $L=1$ the NiBM is a DMN and Theorem~\ref{thm:cut-set-bound}
is the classic cut-set bound.
Alternatively, we may view the DMN as a NiBM with block length $L$ and with
\begin{align} \label{eq:dmn-channel}
P(y_{\set{K}}^L \| x_{\set{K}}^L ) = \prod_{i=1}^L P_{Y_{\set{K}}|X_{\set{K}}}(y_{\set{K},i} | x_{\set{K},i}).
 \end{align}
The weakened bound \eqref{eq:weak} becomes
\begin{align}
& I( X_{\set{S}}^L, 0Y_{\set{S}}^{L-1} \rightarrow Y_{\set{S}^c}^L \| X_{\set{S}^c}^L ) \nonumber \\
& \quad = \sum_{i=1}^L H( Y_{\set{S}^c,i} | X_{\set{S}^c}^i Y_{\set{S}^c}^{i-1} )
    - H( Y_{\set{S}^c,i} | X_{\set{K},i} ) \nonumber \\
& \quad \le \sum_{i=1}^L I( X_{\set{S},i} ; Y_{\set{S}^c,i} | X_{\set{S}^c,i} ).
\label{eq:weak-dmn}
\end{align}
If we choose the code functions as codewords and
\begin{align} \label{eq:cut-set-pdf-dmn}
P(x_{\set{K}}^L) = \prod_{i=1}^L  P(x_{\set{K},i})
\end{align}
then we achieve equality in \eqref{eq:weak-dmn}.
We recover the classic cut-set bound by choosing 
$P(x_{\set{K},i})=P_{\set{X}_{\set{K}}}(x_{\set{K},i})$ for all $i$.

\begin{remark}
Consider a DMN that is time varying in blocks of length $L$, i.e.,
we have a NiBM of length $L$ and
\begin{align} \label{eq:dmn-channel-2}
P(y_{\set{K}}^L \| x_{\set{K}}^L ) = \prod_{i=1}^L P_{Y_{\set{K},i}|X_{\set{K},i}}(y_{\set{K},i} | x_{\set{K},i})
 \end{align}
The cut-set bound of Theorem~\ref{thm:cut-set-bound} may now be computed with independent
inputs as in \eqref{eq:cut-set-pdf-dmn}.
\end{remark}
%

\section{Point-to-Point Channels}
\label{sec:ptp}
Consider a point-to-point channel with input $X^L$ taking on values in $\set{X}^L$,
receiver output $Y^L$ taking on values in $\set{Y}^L$,
and feedback $\tilde{Y}^L$ taking on values in $\set{\tilde{Y}}^L$.
A FDG for $L=2$ and $n=4$ is shown in Fig.~\ref{fig:ptp-iBM-3}. 
\begin{theorem} \label{thm:ptp-C}
The capacity of a point-to-point channel with iBM and block length $L$ is
\begin{align} \label{eq:ptp-C}
 C = \max_{P_{{\mathbf A}^L}} I({\mathbf A}^L ; Y^L)/L
\end{align}
where $P({\mathbf a}^L,y^L,\tilde{y}^L)$ factors as
\begin{align} \label{eq:ptp-pdf}
 P({\mathbf a}^L) 1(x^L \| {\mathbf a}^L,0\tilde{y}^{L-1}) P(y^L,\tilde{y}^L \| x^L ).
\end{align}
\end{theorem}
\begin{proof}
Achievability follows by random coding with a maximizing $P_{{\mathbf A}^L}$.
For example, one may use the steps outlined in~\cite[Sec. VII.B]{Kramer03}.
The converse follows by Theorem~\ref{thm:cut-set-bound}.
\end{proof}

\begin{figure}[t!]
  \centerline{\includegraphics[scale=0.5]{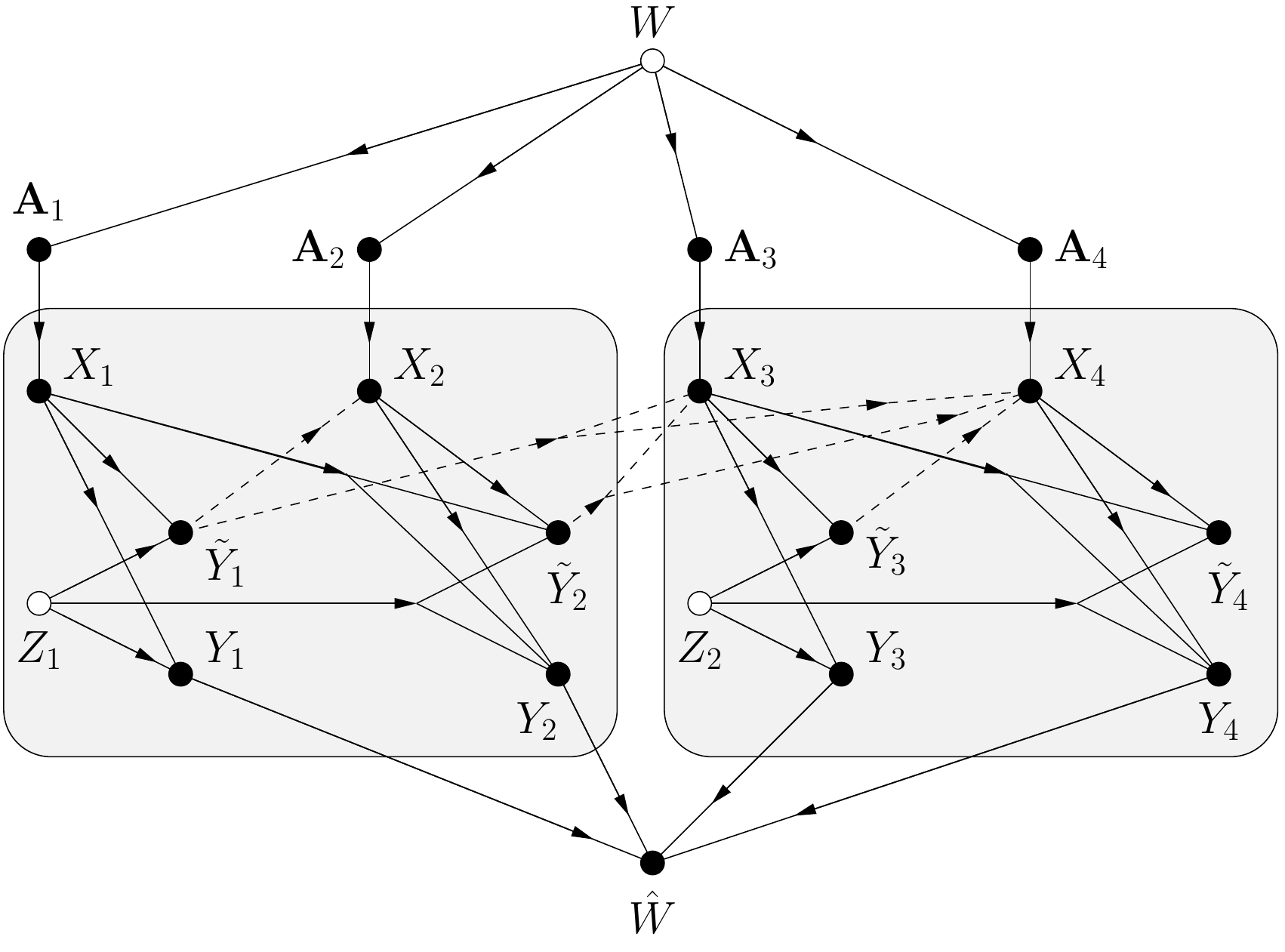}}
  \caption{FDG for a point-to-point channel with iBM and block length $L=2$ and $n=4$ channel uses.}
  \label{fig:ptp-iBM-3}
\end{figure}

\begin{remark}
The distribution \eqref{eq:ptp-pdf} gives
\begin{align}
  I({\mathbf A}^L ; Y^L) = I({\mathbf A}^L \rightarrow Y^L).
\end{align}
\end{remark}
\begin{remark}
The feedback $\tilde{Y}^L$ can be {\em noisy}. 
\end{remark}
\begin{remark}  \label{remark:in-block-fb}
{\em In-block} feedback can increase $C$
but {\em across-block} feedback does not increase $C$.
This statement refines Shannon's classic theorem on feedback capacity \cite[Thm.~6]{Shannon56}.
For instance, in Fig.~\ref{fig:ptp-iBM-3} we can remove the dashed lines
across blocks without changing $C$.
\end{remark}
\begin{remark}  \label{remark:no-fb}
If $\tilde{Y}^L$ is a constant then there is no feedback and we have
\begin{align}
  I({\mathbf A}^L ; Y^L) = I(X^L ; Y^L) = I(X^L \rightarrow Y^L).
\end{align}
The corresponding capacity result is not new, however, since the model is a
special case of a point-to-point channel with vector alphabets.
\end{remark}
\begin{remark}
$I({\mathbf A}^L ; Y^L)$ is concave in $P_{{\mathbf A}^L}$
and the Arimoto-Blahut algorithm~\cite{Arimoto72,Blahut72} can perform the
maximization \eqref{eq:ptp-C}.
\end{remark}

The cardinality $|\set{A}^L|$ is bounded by the channel
alphabets (see \eqref{eq:card-A} and Remark~\ref{remark:auxRV}) and we have
\begin{align} \label{eq:card-ptop}
   \left|\set{A}^L\right| = \prod_{i=1}^L |\set{X}_i|^{\left|\set{\tilde{Y}}^{i-1}\right|} .
\end{align} 
The identity \eqref{eq:card-ptop} means that $|\set{A}^L|$ is double
exponential in $L$ if the alphabet sizes are similar for all $i$. However,
we prove the following theorem by using classic
results~\cite[p.~96]{Gallager68}, \cite[p.~310]{Csiszar81} on bounding set sizes.
\begin{theorem} \label{thm:ptp-cardinality}
The maximum in Theorem~\ref{thm:ptp-C} is achieved by a $P_{{\mathbf A}^L}$
for which $|\supp(P_{{\mathbf A}^L})|$ is at most
\begin{align} \label{eq:ptop-cardinality}
  \min\left( \left|\set{Y}^L\right|,
      \left| \set{X}_1 \right| + \sum_{i=2}^L \left|\set{X}^{i-1}\right| \cdot \left|\set{\tilde{Y}}^{i-1}\right| \cdot
      \left( \left| \set{X}_i \right| -1 \right) \right).
\end{align}
\end{theorem}
\begin{proof}
See Appendix~\ref{app:cardinality-ptp}.
\end{proof}

\begin{remark}
Theorem~\ref{thm:ptp-cardinality} states that $|\supp(P_{{\mathbf A}^L})|$ can be
exponential, and not double exponential, in $L$. Of course, one must still 
determine $\supp(P_{{\mathbf A}^L})$ which can be a high-complexity search problem
for even small $L$.
\end{remark}
\begin{example} \label{example:BSCs}
Consider a binary-alphabet channel with $L=2$ and
\begin{align} \label{eq:ptp-example}
   & Y_1 = X_1, \quad \tilde{Y}_1 = Z, \quad Y_2 = X_2 \oplus Z
\end{align}
where the bit $Z$ has $P_{Z}(1)=\epsilon$.
This is an additive noise channel of the form~\eqref{eq:additive-noise} whose
capacity without feedback is achieved by uniformly-distributed $X^2$
so that
\begin{align} \label{eq:ptp-C-example}
  I(X^2 ; Y^2) = 2 - H_2(\epsilon).
\end{align}  
To compute the feedback capacity, consider the simple bound
\begin{align}
   I({\mathbf A}^2 ; Y^2) & = H(Y^2) - H(Y^2 | {\mathbf A}^2 )
   \le  2 \label{eq:ptp-example-C}
\end{align}
and observe that we achieve equality in \eqref{eq:ptp-example-C} with $X_2=X_2' \oplus Z$
where $X_2'$ is independent of $X_1$, and where $X_1$ and $X_2'$ are
uniformly distributed bits. Feedback thus enlarges the capacity.

We translate this strategy into a code function (here a code tree) distribution.
We label ${\mathbf A}^2$ as $b,b_0b_1$ by which we mean that $X_1=b$,
$X_2=b_0$ if $\tilde{Y}_1=0$, and $X_2=b_1$ if $\tilde{Y}_1=1$. We choose
\begin{align*}
  & P_{{\mathbf A}^2}(0,00)=P_{{\mathbf A}^2}(0,11)=P_{{\mathbf A}^2}(1,00)=P_{{\mathbf A}^2}(1,11)=0 \\
  & P_{{\mathbf A}^2}(0,01)=P_{{\mathbf A}^2}(0,10)=P_{{\mathbf A}^2}(1,01)=P_{{\mathbf A}^2}(1,10)=1/4
\end{align*}
and achieve capacity with four code trees, as predicted by Theorem~\ref{thm:ptp-cardinality}.
\end{example}
\begin{example} \label{ex:weak}
We demonstrate the deficiencies of the weakened bound
based on \eqref{eq:weak}. Suppose the channel is
\begin{align} \label{eq:channel-deficiency}
   & Y_1 = X_1 \oplus Z_1 \oplus Z_2, \quad \tilde{Y}_1 = Z_1, \quad Y_2 = Z_2
\end{align}
where $Z_1$ and $Z_2$ are independent with $P_{Z_1}(1)=\epsilon_1$
and $P_{Z_2}(1)=\epsilon_2$. We achieve the capacity
\begin{align}
   C=(1-H_2(\epsilon_1))/2
\end{align}
by having the receiver compute $Y_1 \oplus Y_2 = X_1 \oplus Z_1$. In fact,
we can achieve capacity by not using the feedback.

For the weakened bound~\eqref{eq:weak2}, observe that \eqref{eq:channel-deficiency}
has the form~\eqref{eq:additive-noise}.
Defining $\epsilon_1*\epsilon_2=\epsilon_1(1-\epsilon_2)+(1-\epsilon_1)\epsilon_2$
and $\set{S}=\{1\}$ we compute (see~\eqref{eq:weak2})
\begin{align}
   H(Z_{\set{S}^c}^L \| 0Z_{\set{S}}^{L-1})
   &  = H(Z_1\oplus Z_2) + H(Z_2| Z_1\oplus Z_2, Z_1) \nonumber \\
   & = H_2(\epsilon_1 * \epsilon_2).
\end{align}
The weakened bound based on~\eqref{eq:weak2} is therefore
\begin{align}
   2C & \le 
    \max_{P_{X^2 \| 0\tilde{Y}_1}} H(Y^2) - H_2(\epsilon_1 * \epsilon_2) \nonumber \\
   & = 1 + H_2(\epsilon_2) - H_2(\epsilon_1 * \epsilon_2) \label{eq:ptp-example2}
\end{align}
with equality if $X_1$ is uniform.
This bound is loose in general, e.g., if $\epsilon_1=1/2$ then $C=0$ but
\eqref{eq:ptp-example2} gives $C\le H_2(\epsilon_2)/2$.
\end{example}

\subsection{Noise-Free Feedback}
\label{subsec:noise-free}
The feedback is {\em noise-free} if $\tilde{Y}^L$ is a causal function of $X^L$ and $Y^L$,
i.e., if $\tilde{Y}_i=f_i(X^i,Y^i)$ for $i=1,2,\ldots,L$. 
The receiver can therefore track, or {\em observe}, the choice
of $X^L$ for each tree ${\mathbf A}^L$. The expression \eqref{eq:ptp-C} simplifies to
\begin{align} \label{eq:ptp-C-2}
 C = \max_{P_{X^L \left\| 0\tilde{Y}^{L-1} \right.}} I(X^L \rightarrow Y^L)/L.
\end{align}
\begin{example}
Consider the additive noise channel \eqref{eq:linear-channel} with
$\ul{Y} = {\mathbf G} \ul{X} + \ul{Z}$ and noise-free feedback. We have
\begin{align}
   I(X^L \rightarrow Y^L) = H(Y^L) - H(Z^L)
   \label{eq:noise-free-additive}
\end{align}
so that computing \eqref{eq:ptp-C-2} reduces to maximizing $H(Y^L)$.
\end{example}

\begin{remark} \label{remark:maximum-entropy}
As in \eqref{eq:noise-free-additive}, one is sometimes interested in maximizing the
output entropy $H(Y^L)$. We observe that for noisy or noise-free feedback we have
\begin{align} 
  \max_{P_{{\mathbf A}^L}} H(Y^L) = \max_{P_{X^L \left\| 0\tilde{Y}^{L-1} \right.}} H(Y^L) .
\end{align}
\end{remark}

\subsection{Block Fading Channels}
\label{subsec:block-fading}
Channels with {\em block fading}~\cite{OzarowShamaiWyner94}
or {\em block interference}~\cite{McElieceStark84} have a
{\em state} $S$ that is memoryless across blocks of length $L$
and whose realization $S=s$ specifies the memoryless channel
in each block. In other words, when $S=s$ we have
\begin{align} \label{eq:bi-channel-s}
P(y^L,\tilde{y}^L \| x^L | s )
= \prod_{i=1}^L P_{Y\tilde{Y}|XS}(y_i,\tilde{y}_i | x_i,s).
 \end{align}
We may view such channels as NiBMs for which $Z=S N^L$ in \eqref{eq:Yfd},
i.e., $Z$ includes the state $S$ and a noise string $N^L$ where
the $N_i$, $i=1,2,\ldots,L$, are statistically independent and identically
distributed. Equation \eqref{eq:Yfd} thus becomes
\begin{align}
   Y_i = f_{t(i)+1}\left( X_{i}, S_{\lceil i/L \rceil} N_i \right) \\
   \tilde{Y}_i = \tilde{f}_{t(i)+1}\left( X_{i}, S_{\lceil i/L \rceil} N_i \right)
\end{align}
for $i=1,2,\ldots,n$.

\subsection{Channels with State Known Causally at the Encoder}
\label{subsec:state}
Shannon's channel with state known causally at the encoder~\cite{Shannon58}
is a point-to-point channel with input and output sequences $X^n$ and $Y^n$,
respectively, and where a state sequence $S^n$ is revealed {\em causally} to the
encoder in the sense that $X_i$ can be a function of $W$ and $S^i$,
$i=1,2,\ldots,n$. The $S_i$, $i=1,2,\ldots,n$, are statistically
independent realizations of a state random variable $S$. The channel outputs
are
\begin{align}
   Y_i = f(X_i,S_i,Z_i)
\end{align}
for some function $f(\cdot)$ where the
$Z_i$, $i=1,2,\ldots,n$, are statistically independent realizations of a noise random
variable $Z$. The FDG is shown in Fig.~\ref{fig:Shannon-iBM-2}.

The channel is usually considered {\em memoryless}. However,
an alternative and insightful interpretation is that this channel has iBM and
block length $L=2$. To see this, observe that Fig.~\ref{fig:Shannon-iBM-2}
is a subgraph of Fig.~\ref{fig:ptp-iBM-3} up to relabeling the nodes.
In other words, in Fig.~\ref{fig:ptp-iBM-3} we choose
$\set{X}_1=\set{Y}_1=\tilde{\set{Y}}_2=\{0\}$ and $\tilde{Y}_1=S$.  
Observe that the ``feedback'' $S$ can be noisy in the sense of Sec.~\ref{subsec:noise-free}.
For the FDG in Fig.~\ref{fig:Shannon-iBM-2} we have renamed ${\bf A}_2$, $\tilde{Y}_1$,
$X_2$, $Y_2$ as ${\bf A}_1$, $S_1$, $X_1$, $Y_1$, respectively, so that
the subscripts enumerate the block. The same random variables without the
block indices are the respective  ${\bf A}$, $S$, $X$, $Y$.
The code functions ${\bf A}$ for this type of problem are sometimes called 
{\em Shannon strategies}~\cite[p.~176]{ElGamal-Kim-11}.

\begin{figure}[t!]
  \centerline{\includegraphics[scale=0.5]{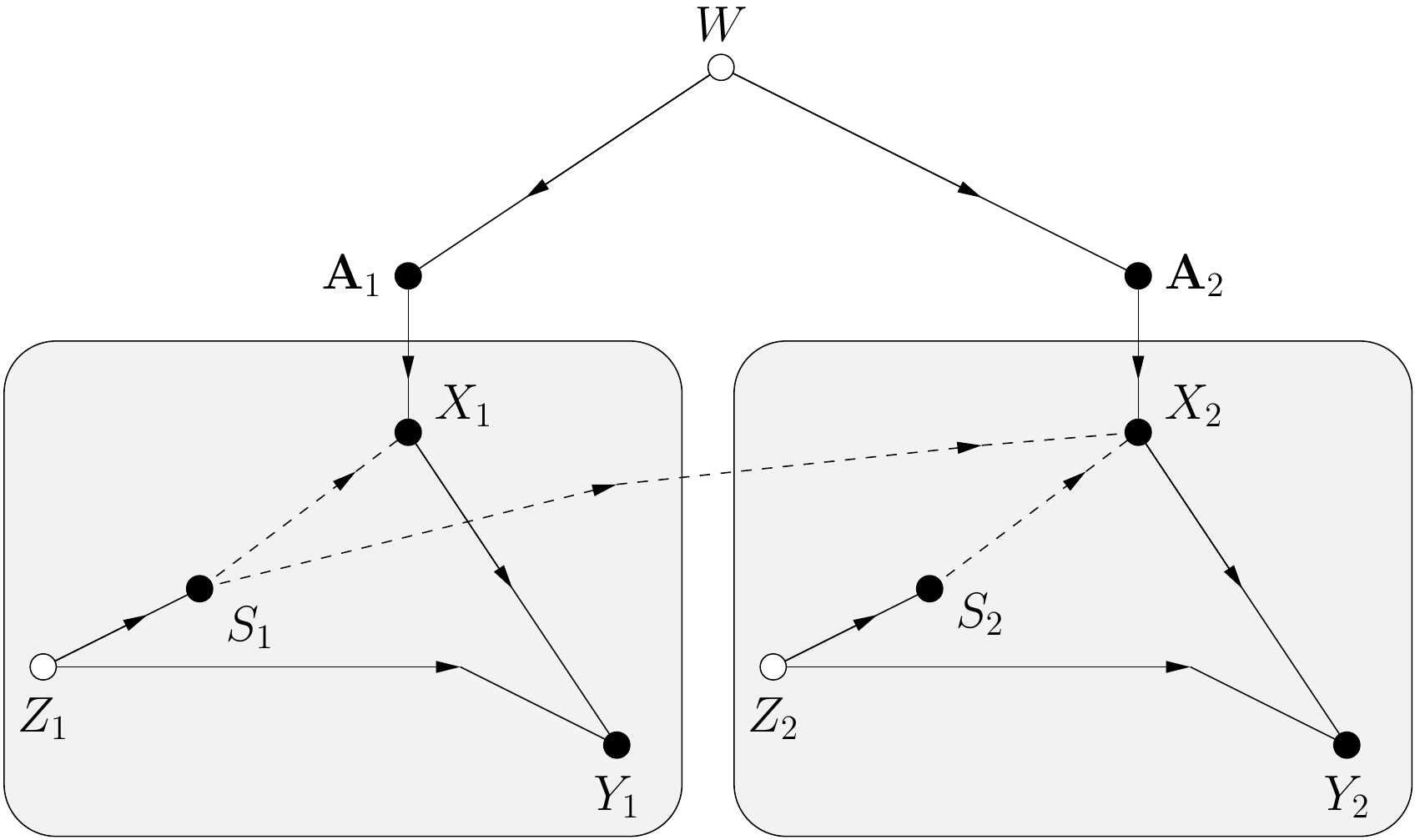}}
  \caption{FDG for a channel with state known causally at the encoder.
  The NiBM has $L=2$. The message estimate $\hat{W}$ is not shown.}
  \label{fig:Shannon-iBM-2}
\end{figure}

The  capacity is given by Theorem~\ref{thm:ptp-C} which here is
\begin{align}
   2C = \max_{P_{{\mathbf A}}} I( {\mathbf A} ; Y ) .
\end{align}
The alphabet size of ${\mathbf A}$ is $|\set{X}|^{|\set{S}|}$ but
\eqref{eq:ptop-cardinality} tells us that
\begin{align} \label{eq:card-state}
   |\supp(P_{{\mathbf A}})| \le \min\left(|\set{Y}|,1+|\set{S}| \cdot (|\set{X}|-1) \right)
\end{align}
suffices. The $|\set{Y}|$ bound is due to Shannon~\cite{Shannon58} and
the second bound was reported in~\cite[Thm.~1]{Farmanbar09}
(see also~\cite[p.~177]{ElGamal-Kim-11}).

\begin{example} \label{example:state}
Suppose that $\set{S}=\set{X}=\{0,1\}$, $\set{Y}=\{0,1,2\}$,
$P_S(0)=1/2$, and
\begin{align}
   Y=X + S
\end{align}
where "$+$" denotes integer addition.
We label the branch-pairs ${\mathbf A}$ as $b_0b_1$, by which we
mean that $X=b_0$ if $S=0$ and $X=b_1$ if $S=1$.
The capacity turns out to be $2C=1$ bit and is achieved with
\begin{align*}
  & P_{{\mathbf A}}(00)=P_{{\mathbf A}}(11)=0 \\
  & P_{{\mathbf A}}(01)=P_{{\mathbf A}}(10)=1/2.
\end{align*}
We thus require at most three code trees, as predicted by \eqref{eq:card-state}.
Moreover, the weakened bound based on \eqref{eq:weak} gives
\begin{align}
  2C \le \max_{P_{X|S}} I(XS;Y) = \log_2(3) \text{ bits.}
\end{align}
A better upper bound follows by giving
$S$ to the receiver to obtain
\begin{align}
  2C \le \max_{P_{X|S}} I(X;Y|S) = 1 \text{ bit.}
\end{align}
\end{example}

\begin{remark} \label{remark:state}
The above construction extends in an obvious way to show that any DMN with
state(s) known causally at the encoder(s) is effectively a NiBM with block length $L=2$.
The cut-set bound \eqref{eq:cut-set-bound} thus applies to these problems.
\end{remark}

\subsection{Channels with Action-Dependent State}
\label{subsec:action}
Weissman's channel with action-dependent state modifies Shannon's model and
lets the transmitter influence the state~\cite{Weissman10}.
In other words, at time $i$ the transmitter can choose a letter $B_i$ as
a function of $W$ and $S^{i-1}$ and the next state is
\begin{align}
   S_i = g(B_i,Z_i)
\end{align}
for some function $g(\cdot)$. The FDG is shown in Fig.~\ref{fig:action-iBM-2}.
Observe that $Z_i$ could be a random vector so that the noise
influencing $S_i$ and $Y_i$ is statistically independent.

This channel is again usually considered memoryless. However,
we interpret the channel as having iBM and block
length $L=2$, since Fig.~\ref{fig:action-iBM-2} is a subgraph of
Fig.~\ref{fig:ptp-iBM-3} up to relabeling the nodes.
More precisely, in Fig.~\ref{fig:ptp-iBM-3} we choose
$\set{Y}_1=\tilde{\set{Y}}_2=\{0\}$ and $\tilde{Y}_1=S$.  
For the FDG in Fig.~\ref{fig:action-iBM-2}
we have renamed $X_1$, $\tilde{Y}_1$,
$X_2$, $Y_2$ as $B_1$, $S_1$, $X_1$, $Y_1$, respectively, so that
the subscripts enumerate the block. The same random
variables without the block indices are the respective  $B$, $S$, $X$, $Y$.
Theorem~\ref{thm:ptp-C} gives the capacity
\begin{align}
   2C & = \max_{P_{{\mathbf A}^2}} I({\mathbf A}^2 ; Y) 
   = \max_{P_{B{\mathbf A}_2}} I(B{\mathbf A}_2 ; Y)
   \label{eq:action-cap}
\end{align}
and Theorem~\ref{thm:ptp-cardinality}
gives
\begin{align}
   |\supp(P_{B{\mathbf A}_2})|
   \le \min\left(|\set{Y}|, |\set{B}| + |\set{B}| \, |\set{S}| \, (|\set{X}|-1) \right).
   \label{eq:action-card}
\end{align}

\begin{figure}[t!]
  \centerline{\includegraphics[scale=0.5]{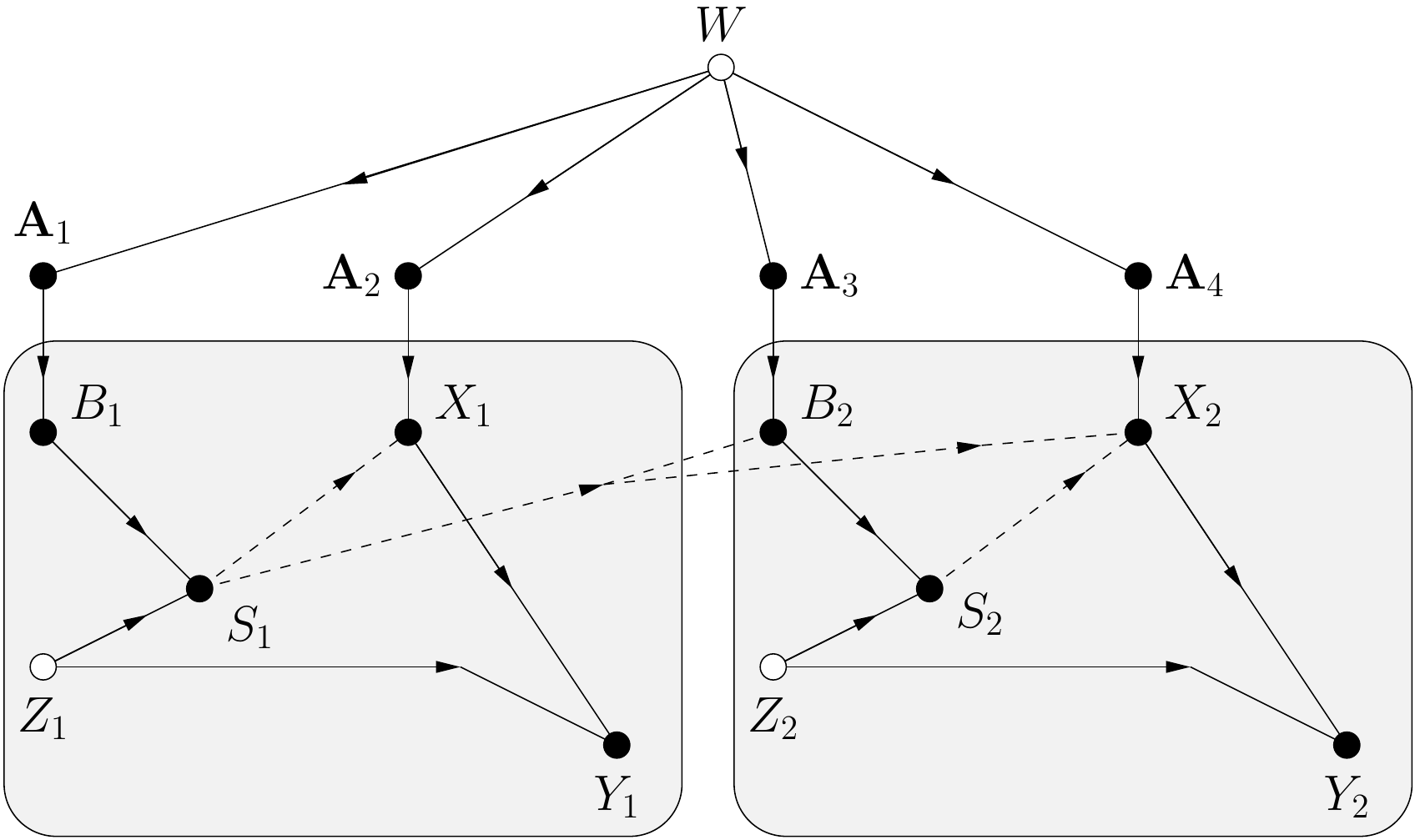}}
  \caption{FDG for a channel with action-dependent state known causally at the encoder.
  The NiBM has $L=2$ and the actions are $B_1$ and $B_2$.}
  \label{fig:action-iBM-2}
\end{figure}

\begin{remark}
The expression \eqref{eq:action-cap} is the same as in~\cite[Thm.~2]{Weissman10}
because $U$ plays the role of $B{\mathbf A}_2$.
\end{remark}
\begin{remark}
The constraint \eqref{eq:action-card} is slightly stronger than that
in~\cite[Thm.~2]{Weissman10}.
\end{remark}
\begin{remark}
The model in Fig.~\ref{fig:action-iBM-2} seems more general than
in~\cite{Weissman10} because $Z$ may influence both $S$ and $Y$.
However, the associations described in~\cite[p.~5405]{Weissman10} show
that the original model includes more problems than apparent at first glance
(see also comments in~\cite[Sec.~VII]{Weissman10}).
\end{remark}
\begin{remark}
The model in Fig.~\ref{fig:action-iBM-2} may seem
different than in~\cite{Weissman10} because $S$ may influence future actions
as well as the present and future $X$. However, across-block feedback does
not increase capacity (see Remark~\ref{remark:in-block-fb}) so we may remove
the $S$-to-$B$ functional dependence without affecting capacity
(see also comments in~\cite[Sec.~VII]{Weissman10} concerning feedback).
\end{remark}
\begin{remark}
We may add functional dependence from $B$ to $Y$ without changing 
the capacity expression. Similar comments are made in~\cite[p.~5398 and Sec.~VII]{Weissman10}.
\end{remark}
\begin{example}
Consider a channel with a rewrite option~\cite[Sec.~V.A]{Weissman10} which means
that the $B$-to-$S$ and $X$-to-$Y$ channels are effectively the same.
At time $i=1$ the encoder ``writes" on the $B$-to-$S$ channel.
At time $i=2$, if the encoder is happy with the outcome $S$ then it sends a
no-rewrite symbol $N$ which means that $Y=S$. But if the encoder
is unhappy with $S$ then it ``rewrites" a symbol on the $X$-to-$Y$ channel.

We have $\set{X}=\set{B} \cup \{N\}$, $\set{S}=\set{Y}$, and
the bound \eqref{eq:action-card} is $|\supp(P_{B{\mathbf A}_2})|\le|\set{Y}|$.
For example, suppose the $B$-to-$S$ channel is a binary symmetric channel (BSC)
with crossover probability
$\delta$, $0\le\delta\le1/2$ (see~\cite{Weissman10}). We label $B{\mathbf A}_2$ as
$b,b_0b_1$ by which we mean that $B=b$, $X=b_0$ if $S=0$, and
$X=b_1$ if $S=1$. We have $|\set{Y}|=2$ and achieve
$C=I( B{\mathbf A}_2 ; Y)=1-H_2(\delta^2)$ by choosing
\begin{align*}
  & P_{B{\mathbf A}_2}(0,N0)=P_{B{\mathbf A}_2}(1,1N)=1/2.
\end{align*}
We require only two code trees, as predicted by \eqref{eq:action-card}.
\end{example}
\begin{remark}
Multiple rewrites are modeled by increasing $L$.
\end{remark}
%

\section{Multiuser Channels}
\label{sec:multiuser}
%
\subsection{Multiaccess Channels}
\label{subsec:MACs}
Consider a two-user (three-terminal) MAC with iBM and
with inputs $X_1^L$, $X_2^L$, and outputs $Y_1^L$, $Y_2^L$, $Y_3^L$.
Node 3 is the receiver and the variables $X_3^L$ should be considered constants.
The FDG for $L=2$ and $n=4$ is the same as
Fig.~\ref{fig:twc-block-memory-2} except that the variables $Y_{3i}$, $i=1,2,3,4$,
are missing in Fig.~\ref{fig:twc-block-memory-2}.
The cut-set bound of Theorem~\ref{thm:cut-set-bound} is
\begin{align} \label{eq:mac-fb}
 \bigcup_{P_{{\mathbf A}_1^L {\mathbf A}_2^L}} \left\{ (R_1,R_2):
 \begin{array}{l}
  0 \le R_1, \; 0 \le R_2 \\
  R_1 \le I({\mathbf A}_1^L ; Y_3^L | {\mathbf A}_2^L )/L \\
  R_2 \le I({\mathbf A}_2^L ; Y_3^L | {\mathbf A}_1^L )/L \\
  R_1+R_2 \le I({\mathbf A}_1^L {\mathbf A}_2^L ; Y_3^L )/L
  \end{array} \right\} .
\end{align}
If there is no feedback, then $Y_1^L$ and $Y_2^L$ can be considered
constants. The resulting cut-set bound can be strengthened in the usual way to
become
\begin{align} \label{eq:mac}
 \bigcup \left\{ (R_1,R_2):
 \begin{array}{l}
  0 \le R_1, \; 0 \le R_2 \\
  R_1 \le I(X_1^L ; Y_3^L | X_2^L T)/L \\
  R_2 \le I(X_2^L ; Y_3^L | X_1^L T)/L \\
  R_1+R_2 \le I(X_1^L X_2^L ; Y_3^L | T)/L
    \end{array} \right\}
\end{align}
where the union is over distributions such that $X_1^L - T - X_2^L$ forms a Markov chain
($T$ is the usual time-sharing random variable). This modified cut-set bound is the
capacity region without feedback. The result is not new, however, since the model is a
special case of a classic MAC with vector alphabets.

\begin{remark} \label{remark:mac-state}
MACs with state known causally at the encoders were treated
in~\cite[Sec.~IV]{SigurjonssonKim05}. As pointed out in Remark~\ref{remark:state},
such channels are NiBMs with block length $L=2$. For example, the outer bound of
Theorem~3 in~\cite[Sec.~IV]{SigurjonssonKim05} is the same as the
cut-set bound of Theorem~\ref{thm:cut-set-bound}.
\end{remark}

\subsection{Multiaccess Channels with Feedback}
\label{subsec:MAC-FB}
Several capacity results for DMNs generalize to problems with iBM.
For example, consider Willems' result~\cite{Willems:82} that the Cover-Leung
region~\cite{Cover81} is $\set{C}$ for full feedback ($Y_1=Y_2=Y_3=Y$) and
where one channel input, say $X_1$, is a function of $Y$ and $X_2$.
A natural generalization to MACs with iBM is to consider full feedback
($Y_{1,i}=Y_{2,i}=Y_{3,i}=Y_i$) and require $X_{1,i}=f_{i}(X_2^i,Y^i)$ for $i=1,2,\ldots,L$.
A MAC of this type is the binary adder channel (BAC) with $\{0,1\}$ input
alphabets and the integer-addition output
\begin{align} \label{eq:modulo-additive}
   \ul{Y} = {\mathbf G}_1 \ul{X}_1 + {\mathbf G}_2 \ul{X}_2
\end{align}
where ${\mathbf G}_1$ and  ${\mathbf G}_2$ are lower-triangular matrices with $\{0,1\}$ entries,
and where ${\mathbf G}_1$  has ones on the diagonal.

\begin{theorem} \label{thm:mac-fb}
The capacity region of a MAC with iBM and full feedback and where $X_{1,i}=f_{i}(X_2^i,Y^i)$
for all $i$ is
\begin{align} \label{eq:mac-fb-3}
 \bigcup \left\{ (R_1,R_2):
 \begin{array}{l}
  0 \le R_1, \; 0 \le R_2 \\
  R_1 \le I({\mathbf A}_1^L ; Y^L | {\mathbf A}_2^L \, V)/L \\
  R_2 \le I({\mathbf A}_2^L ; Y^L | {\mathbf A}_1^L \, V)/L \\
  R_1+R_2 \le I({\mathbf A}_1^L {\mathbf A}_2^L ; Y^L)/L 
 \end{array} \right\}
\end{align}
where the union is over distributions that factor as
\begin{align} \label{eq:mac-fb-pdf-2}
  P(v) \left[  \prod_{k=1}^2 P({\mathbf a}_k^L|v) 1(x_k^L \| {\mathbf a}_k^L,0y^{L-1}) \right]
  P(y^L \| x_1^L, x_2^L ).
\end{align}
A cardinality bound on $V$ is $|\set{V}| \le \left|\set{Y}^L\right|+2$.
\end{theorem}
\begin{proof}
The proof mimics that in~\cite{Willems:82} and is given in Appendix~\ref{app:mac-fb}.
\end{proof}

\begin{proposition} \label{prop:mac-fb}
An alternative way of writing \eqref{eq:mac-fb-3}-\eqref{eq:mac-fb-pdf-2} is
\begin{align} \label{eq:mac-fb-2}
 \bigcup \left\{ (R_1,R_2):
 \begin{array}{l}
  0 \le R_1, \; 0 \le R_2 \\
  R_1 \le I(X_1^L \rightarrow Y^L \| X_2^L \, | \, V)/L \\
  R_2 \le I(X_2^L \rightarrow Y^L \| X_1^L \, | \, V)/L \\
  R_1+R_2 \le I(X_1^L X_2^L \rightarrow Y^L)/L
 \end{array} \right\}
\end{align}
where the union is over distributions that factor as
\begin{align} \label{eq:mac-fb-pdf}
  P(v) \left[  \prod_{k=1}^2 P(x_k^L \| 0y^{L-1}|v) \right] P(y^L \| x_1^L, x_2^L ).
\end{align}
Note that one conditions on $V$ for all times.
\end{proposition}
\begin{proof}
Consider the distribution \eqref{eq:mac-fb-pdf-2}.
The chains
\begin{align}
& {\mathbf A}_1^L - V X_1^i Y^{i-1}-Y_i \\
& {\mathbf A}_2^L - V X_2^i Y^{i-1}-Y_i \\
& {\mathbf A}_1^L {\mathbf A}_2^L - V X_1^i X_2^i Y^{i-1}-Y_i \label{eq:Markov-3}
\end{align}
are Markov so that
\begin{align}
& I({\mathbf A}_1^L ; Y^L | {\mathbf A}_2^L \, V) \nonumber \\
& = \sum_{i=1}^L H( Y_i | {\mathbf A}_2^L Y^{i-1} X_2^i \, V) - H( Y_i | {\mathbf A}_1^L {\mathbf A}_2^L Y^{i-1} X_1^i X_2^i \, V)
 \nonumber \\
 & = I(X_1^L \rightarrow Y^L \| X_2^L \, | \, V).
\end{align}
and similarly
\begin{align}
& I({\mathbf A}_2^L ; Y^L | {\mathbf A}_1^L \, V) = I(X_2^L \rightarrow Y^L \| X_1^L \, | \, V) \\
& I({\mathbf A}_1^L {\mathbf A}_2^L ; Y^L ) = I(X_1^L X_2^L \rightarrow Y^L ).
\end{align}
The distribution \eqref{eq:mac-fb-pdf} follows from \eqref{eq:mac-fb-pdf-2}.
\end{proof}
%

\subsection{Broadcast Channels}
\label{subsec:BCs}
Consider a two-user (three terminal) BC with iBM. We label the transmitter inputs and outputs
as $X^L$ and $Y^L$, respectively, and the receiver outputs as $Y_1^L$ and $Y_2^L$.
Suppose there are only dedicated messages and no common message.
The cut-set bound of Theorem~\ref{thm:cut-set-bound} is
\begin{align} \label{eq:bc-fb}
 \bigcup_{P_{{\mathbf A}^L}} \left\{ (R_1,R_2):
 \begin{array}{l}
  0 \le R_1 \le I({\mathbf A}^L ; Y_1^L )/L \\
  0 \le R_2 \le I({\mathbf A}^L ; Y_2^L )/L \\
  R_1+R_2 \le I({\mathbf A}^L ; Y_1^L Y_2^L )/L
  \end{array} \right\} .
\end{align}

An achievable region follows by extending Marton's region as in~\cite[Lemma~2]{Kramer03}:\
the non-negative rate pair $(R_1,R_2)$ is achievable if it satisfies
\begin{align} \label{eq:bc-marton}
 \begin{array}{l}
   LR_1 \le I(T U_1 ; Y_1^L ) \\
   LR_2 \le I(T U_2 ; Y_2^L ) \\
   L(R_1+R_2) \le \min\left( I(T;Y_1^L),I(T;Y_2^L) \right)  \\
   \qquad + \, I(U_1 ; Y_1^L | T) + I(U_2 ; Y_2^L|T) - I(U_1 ; U_2 | T)
 \end{array}
\end{align}
for some auxiliary random variables $TU_1U_2$ for which the joint
distribution of the random variables factors as
\begin{align} \label{eq:bc-fb-pdf}
  P(t,u_1,u_2) 1(x^L \| 0y^{L-1} | t,u_1,u_2) P(y_1^L,y_2^L \| x^L ).
\end{align}
Marton's region is known to be the same as \eqref{eq:bc-fb} for $L=1$ and {\em deterministic}
broadcast channels. For $L>1$, suppose that $Y_{1,i}$ and
$Y_{2,i}$ are functions of $X^i$ for all $i$. We may choose $T=0$, $U_1=Y_1^L$, and
$U_2=Y_2^L$ without violating the Markov condition \eqref{eq:bc-fb-pdf} and achieve
\begin{align} \label{eq:bc-fb-deterministic}
 \bigcup_{P_{X^L}} \left\{ (R_1,R_2):
 \begin{array}{l}
  0 \le R_1 \le H(Y_1^L )/L \\
  0 \le R_2 \le H(Y_2^L )/L \\
  R_1+R_2 \le H(Y_1^L Y_2^L )/L
  \end{array} \right\} .
\end{align}
The cut-set region \eqref{eq:bc-fb} is the same as \eqref{eq:bc-fb-deterministic}, and
therefore \eqref{eq:bc-fb-deterministic} is $\set{C}$. In fact, feedback does not increase
capacity because the transmitter knows, and controls, the channel outputs.

\begin{remark} \label{remark:bc-state}
The capacity region of a physically degraded BC with two receivers and state
known causally at the encoder was derived in~\cite[Sec.~II]{SigurjonssonKim05}. 
Such channels are NiBMs with block length $L=2$, see Remark~\ref{remark:state}.
The cut-set bound of Theorem~\ref{thm:cut-set-bound} is loose but the capacity region
is achieved by using the coding method described above. In particular,
we choose $U_2$ in \eqref{eq:bc-marton}-\eqref{eq:bc-fb-pdf} to be a constant and
recover the achievability part of Theorem~1 of~\cite[Sec.~II]{SigurjonssonKim05}.
\end{remark}

\subsection{Interference Channels}
\label{subsec:ICs}
The cut-set bound is often not so interesting for BCs or interference channels (ICs)
with $L=1$ because better capacity bounds exist. The same will be
true for $L>1$. On the other hand, studying extensions of existing bounds and
achievable regions is interesting, e.g., extensions of the 
Han-Kobayashi region~\cite{HK81} to $L>1$.
It may also be interesting to study interference alignment~\cite{CadambeJafar08,Maddah-Ali08}
and interference focusing~\cite{Ghozlan11} for NiBMs.

\section{Relay Networks}
\label{sec:relay-networks}
Causal relay networks~\cite{Baik11} and generalized networks~\cite{FongIT:12}
effectively extend relay networks with delays~\cite{ElGamal:07} in the sense that for
every relay network with delays there is a causal relay network having the same capacity
region. Furthermore, causal relay networks and generalized networks are special NiBMs.
This section focuses on relay networks with iBM and applies
Theorem~\ref{thm:cut-set-bound} to this class of problems.

\subsection{Relay Channels}
\label{subsec:RCs}
%
\begin{figure*}[t!]
  \centerline{\includegraphics[scale=0.5]{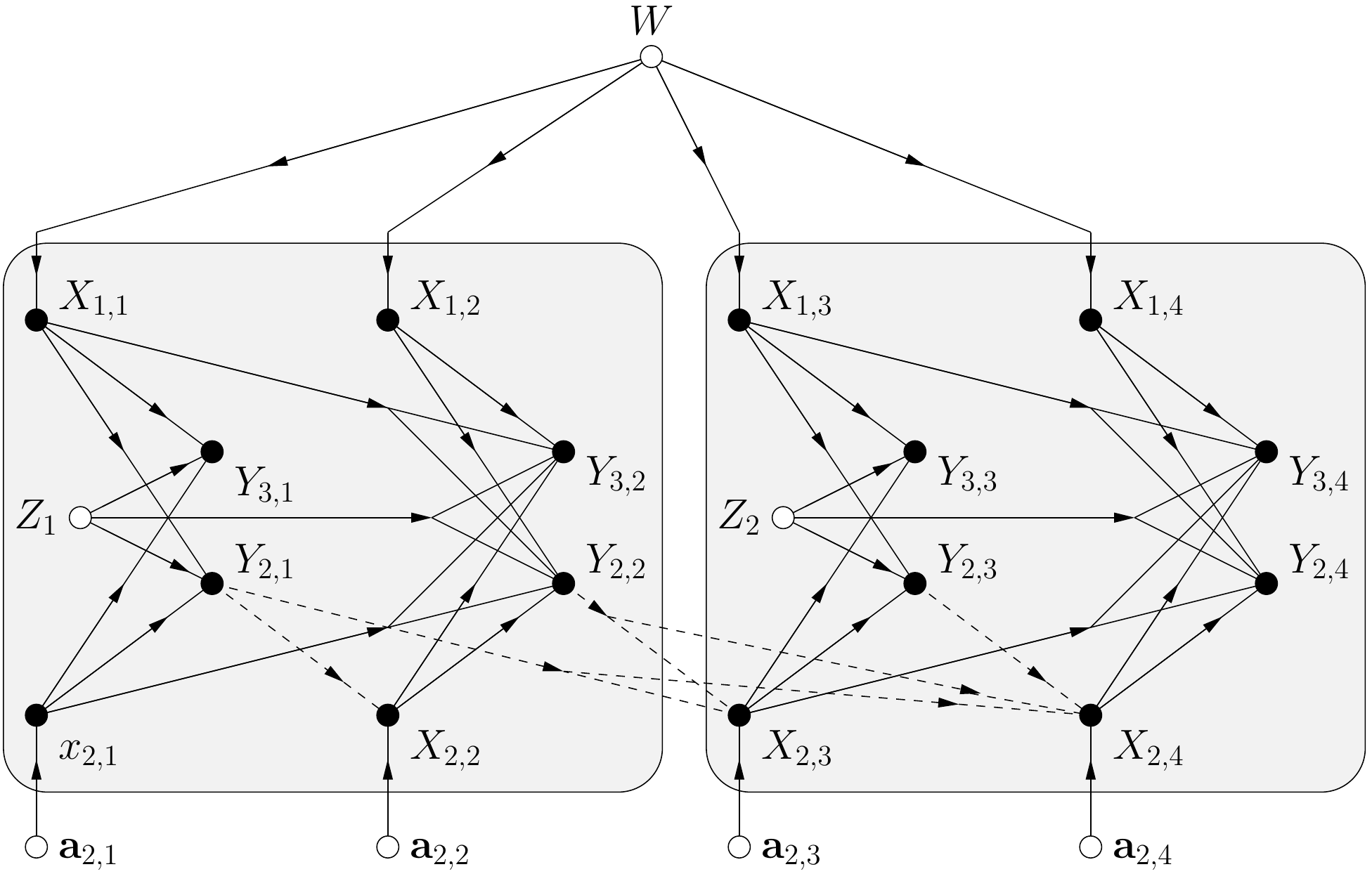}}
  \caption{FDG for a RC with iBM and block length $L=2$.}
  \label{fig:rc-iBM-2}
\end{figure*}
Consider a three-node relay channel (RC) with iBM and source inputs $X_1^L$,
relay inputs $X_2^L$ and outputs $Y_2^L$, and destination outputs $Y_3^L$.
The RC is a special case of the MAC in Sec.~\ref{subsec:MACs} where
node~2 (the relay) has no message and node~1 (the source) has no feedback.
A FDG for $L=2$ and $n=4$ is shown in Fig.~\ref{fig:rc-iBM-2}.
The cut-set bound of Theorem~\ref{thm:cut-set-bound} is
\begin{align} \label{eq:rc-fb}
  L C \le \max \min\left(
   I( X_1^L ; Y_2^L Y_3^L | {\mathbf A}_2^L ), I(X_1^L {\mathbf A}_2^L ; Y_3^L )
   \right)
\end{align}
where the maximization is over $P_{X_1^L{\mathbf A}_2^L}$.

We list several classic coding strategies~\cite{Cover79,Kramer05}. The achievable
rates follow by standard random coding arguments
(see~\cite[Sec.~VI]{Kramer03}).

\begin{itemize}
\item {\em Decode-forward (DF)} achieves rates $R$ satisfying
\begin{align} \label{eq:rc-fb-df}
  L R = \max \min\left(
   I( X_1^L ; Y_2^L | {\mathbf A}_2^L ), I(X_1^L {\mathbf A}_2^L ; Y_3^L )
   \right)
\end{align}
where the maximization is over $P_{X_1^L{\mathbf A}_2^L}$ and
where the joint distribution factors as
\begin{align} \label{eq:rc-fb-df-pdf}
  P(x_1^L,{\mathbf a}_2^L) \, 1(x_2^L \| {\mathbf a}_2^L,0y_2^{L-1}) \, P(y_2^L,y_3^L \| x_1^L,x_2^L ).
\end{align}

\item {\em Partial decode-forward (PDF)} achieves $R$ satisfying
\begin{align} \label{eq:rc-fb-pdf}
   L R & = \max\min\left( I( U ; Y_2^L | {\mathbf A}_2^L ) + I( X_1^L ; Y_3^L | {\mathbf A}_2^L U), \right. \nonumber \\
   & \qquad \qquad \quad \left. I(X_1^L {\mathbf A}_2^L ; Y_3^L ) \right)
\end{align}
where the maximization is over $P_{U X_1^L{\mathbf A}_2^L}$ and
where the joint distribution factors as
\begin{align} \label{eq:rc-fb-pdf-pdf}
  P(u,x_1^L,{\mathbf a}_2^L) 1(x_2^L \| {\mathbf a}_2^L,0y_2^{L-1})  P(y_2^L,y_3^L \| x_1^L,x_2^L ).
\end{align}
The rate \eqref{eq:rc-fb-pdf} generalizes \cite[Prop.~5]{ElGamal:07}.

\item {\em Compress-foward (CF)} achieves $R$ satisfying
\begin{align} \label{eq:rc-fb-cf}
   & L R = \max\min\left( I( X_1^L ; \hat{Y}_2^L Y_3^L | {\mathbf A}_2^L T), \right. \nonumber \\
   & \; \left. I(X_1^L {\mathbf A}_2^L ; Y_3^L | T) - I( Y_2^L ; \hat{Y}_2^L | X_1^L {\mathbf A}_2^L Y_3^L T)  \right)
\end{align}
where the maximization is over joint distributions that factor as
\begin{align} \label{eq:rc-fb-cf-pdf}
  & P(t) \, P(x_1^L|t) \, P({\mathbf a}_2^L|t) \, 1(x_2^L \| {\mathbf a}_2^L,0y_2^{L-1}) \nonumber \\
  & \; \cdot P(\hat{y}_2^L | {\mathbf a}_2^L,y_2^L,t ) \, P(y_2^L,y_3^L \| x_1^L,x_2^L ).
\end{align}
\end{itemize}

\begin{example}
Remark~\ref{remark:channel} states that we can view the
channel as being $P(y_2^L,y_3^L | x_1^L,{\mathbf a}_2^L)$.
The RC is {\em physically degraded} if
\begin{align}
   X_1^L - {\mathbf A}_2^L Y_2^L - Y_3^L
\end{align}
forms a Markov chain so that $I(X_1^L; Y_3^L | {\mathbf A}_2^L Y_2^L)=0$. The DF rate
\eqref{eq:rc-fb-df} thus matches~\eqref{eq:rc-fb}. This capacity result
generalizes \cite[Prop.~6]{ElGamal:07}.
\end{example}

\begin{example}
The RC is {\em reversely} physically degraded if
\begin{align}
   X_1^L - {\mathbf A}_2^L Y_3^L - Y_2^L
\end{align}
forms a Markov chain so that $I(X_1^L; Y_2^L | {\mathbf A}_2^L Y_3^L)=0$.
The cut-set bound \eqref{eq:rc-fb} reduces to
\begin{align} \label{eq:rc-reverse-degraded}
  L C \le \max_{{\mathbf a}_2^L} \max_{P_{X_1^L}}
  I( X_1^L ; Y_3^L | {\mathbf A}_2^L = {\mathbf a}_2^L ) .
\end{align}
The rate on the right-hand side of \eqref{eq:rc-reverse-degraded} is achieved
by random coding with ${\mathbf A}_2^L={\mathbf a}_2^L$.
\end{example}

\begin{remark}
Physically degraded RCs with state known causally at the encoder
are treated in~\cite[Sec.~III]{SigurjonssonKim05}.  Such channels are NiBMs
with block length $L=2$ (see Remark~\ref{remark:state}) and Theorem~\ref{thm:cut-set-bound}
gives the converse for~\cite[Thm.~2]{SigurjonssonKim05}.
However, these channels are not treated 
in this section because the source node receives the channel state as ``feedback''.
\end{remark}

\begin{example}
Suppose the RC is semi-deterministic in the sense that
$Y_{2,i}=f_i(X_1^i,X_2^i)$ for $i=1,2,\ldots,L$. We may choose $U=Y_2^L$ and
\eqref{eq:rc-fb-pdf} becomes the cut-set bound \eqref{eq:rc-fb}.
This capacity result generalizes \cite[Prop.~7]{ElGamal:07}.
\end{example}

\begin{example} \label{ex:CFexample}
Suppose the RC is semi-deterministic in the (more general) sense that 
$Y_{2,i}=f_i(X_1^i,X_2^i,Y_3^i)$ for $i=1,2,\ldots,L$. Consider \eqref{eq:rc-fb-cf}
for which we have
\begin{align}
   I( Y_2^L ; \hat{Y}_2^L | X_1^L {\mathbf A}_2^L Y_3^L T)=0.
\end{align}
We choose $T$ as a constant and $\hat{Y}_2^L=Y_2^L$ so that \eqref{eq:rc-fb-cf}
is the right-hand side of \eqref{eq:rc-fb} but with independent $X_1^L$ and ${\mathbf A}_2^L$.
\end{example}

\begin{example}
A special case of Example~\ref{ex:CFexample} is where $Y_{2,i}=f_i(X_1^i,Y_3^i)$ and 
there is a separate channel with iBM and capacity $R_0$ from the relay to the destination
(see~\cite{Kim:08}). The best $X_1^L$ and ${\mathbf A}_2^L$ are independent so the choice
$\hat{Y}_2^L=Y_2^L$ lets CF achieve the cut-set bound \eqref{eq:rc-fb}.
\end{example}

\subsection{Relays without Delay}
\label{subsec:delay}
A relay {\em without} delay~\cite{ElGamal:07} has source input $X_1$, relay input
$X_2$ and output $Y_2$, and destination output $Y_3$. The channel is
\begin{align} \label{eq:rwod}
   P(y_2 | x_1) \cdot P(y_3 | x_1,x_2,y_2)
\end{align}
and the FDG for two channel uses is shown in Fig.~\ref{fig:rn-delay-iBM-2}.

This channel is usually considered {\em memoryless}. However, we can
model the channel as a RC with iBM and block length $L=2$ and
where $\set{X}_{2,1}=\set{Y}_{3,1}=\set{Y}_{2,2}=\set{X}_{1,2}=\{0\}$. The channel is
therefore
\begin{align} \label{eq:rwod-time}
   P(y_2^2,y_3^2 \| x_1^2,x_2^2) = P(y_{2,1} | x_{1,1}) \cdot P(y_{3,2} | x_{1,1},x_{2,2},y_{2,1})
\end{align}
as long as $x_{2,1}=y_{3,1}=y_{2,2}=x_{1,2}=0$. Note that every node has at most
one channel input and output in each block. We can thus remove the time indices
and \eqref{eq:rwod-time} becomes \eqref{eq:rwod}. Observe that Fig.~\ref{fig:rn-delay-iBM-2}.
is a subgraph of Fig.~\ref{fig:rc-iBM-2} up to relabeling the nodes.

\begin{figure}[t!]
  \centerline{\includegraphics[scale=0.5]{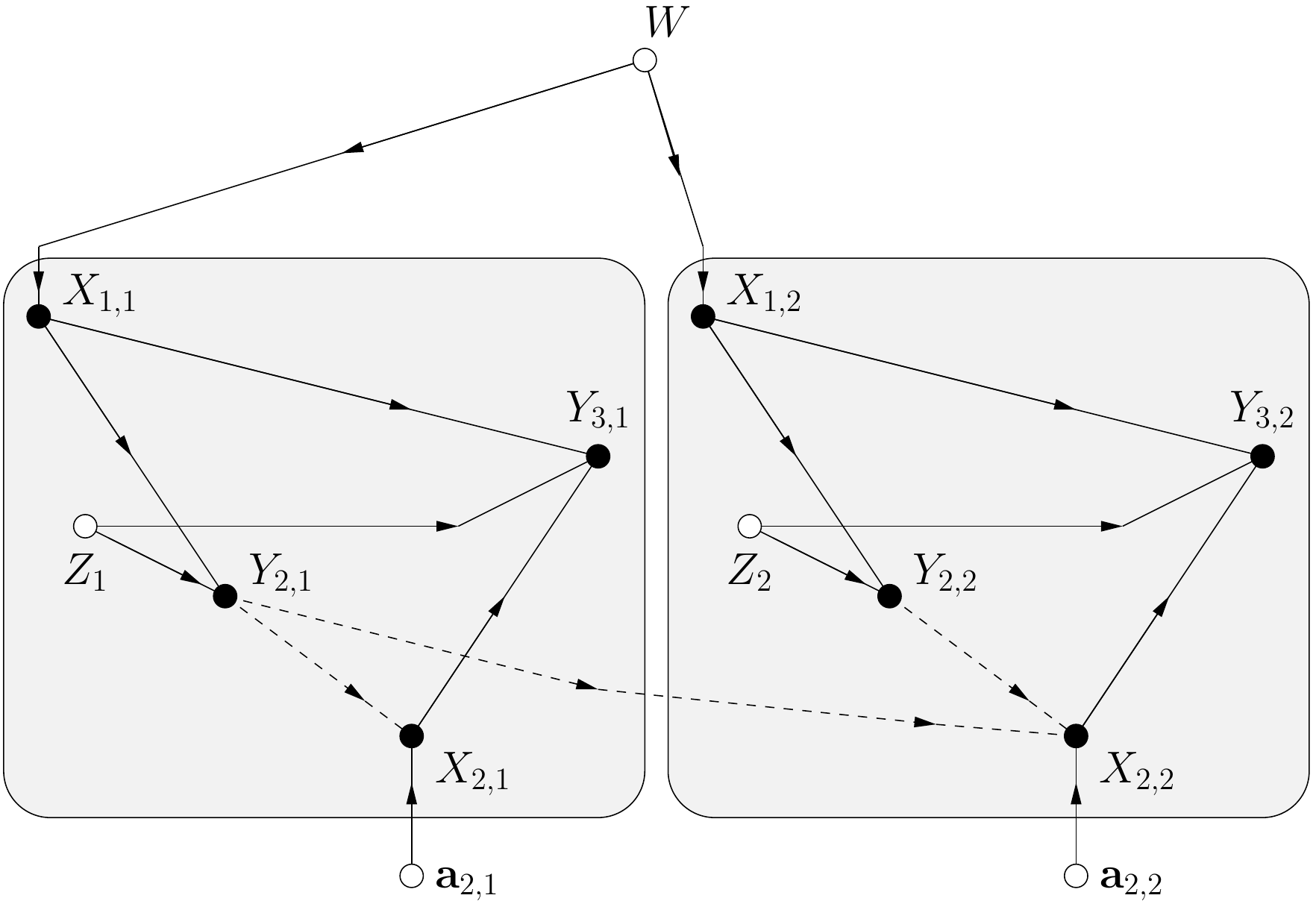}}
  \caption{FDG for a RC when the relay has no delay. The channel is a NiBM with block length $L=2$.}
  \label{fig:rn-delay-iBM-2}
\end{figure}

We apply the cut-set bound \eqref{eq:rc-fb} and remove the time indices to obtain
\begin{align} \label{eq:rc-fb-0}
  2 C \le \max \min\left(
   I( X_{1} ; Y_{2} Y_{3} | {\mathbf A}_{2} ), I(X_{1} {\mathbf A}_{2} ; Y_{3} )
   \right)
\end{align}
where the maximization in \eqref{eq:rc-fb-0} is over $P_{X_{1}{\mathbf A}_{2}}$ and
$|\set{A}_{2}|=|\set{X}_{2}|^{|\set{Y}_{2}|}$.
In fact, \eqref{eq:rc-fb-0} combined with this cardinality  constraint
is attributed to Willems' in \cite[p.~3419]{ElGamal:07}.
We show in Appendix~\ref{app:cardinality-rc-0} that one can choose
\begin{align} \label{eq:rc-fb-0-cardinality}
|\supp(P_{{\mathbf A}_{2}})| \le \min\left(|\set{Y}_{3}|+1, |\set{X}_{1}| \cdot |\set{X}_{2}|+1 \right).
\end{align}

\begin{remark} \label{remark:bound}
The cut-set bound in~\cite[Thm. 2]{ElGamal:07} is the same as \eqref{eq:rc-fb-0}
except that the maximization is different. The bound of~\cite[Thm. 2]{ElGamal:07}
requires
\begin{align} \label{eq:EG07-Thm2}
X_{2} = f({\mathbf A}_{2},Y_{2})
\end{align}
for some function $f(\cdot)$ and one optimizes over all $f(\cdot)$ and
$P_{X_{1}{\mathbf A}_{2}}$ such that
$|\set{A}_{2}| \le |\set{X}_{1}| \cdot |\set{X}_{2}|+1$.

We claim that the formulation \eqref{eq:rc-fb-0} combined with \eqref{eq:rc-fb-0-cardinality}
is better than~\cite[Thm. 2]{ElGamal:07} in the sense that the former
has a smaller search space in general. Observe that \eqref{eq:rc-fb-0}-\eqref{eq:rc-fb-0-cardinality} 
requires optimizing $P_{X_{1}{\mathbf A}_{2}}$ by considering at most
\begin{align}
N_A=\min(|\set{Y}_{3}|+1,|\set{X}_{1}| \cdot |\set{X}_{2}|+1)
\end{align}
out of $|\set{X}_{2}|^{|\set{Y}_{2}|}$
code functions. We must therefore perform at most
\begin{align}
\binom{|\set{X}_{2}|^{|\set{Y}_{2}|}}{N_A}
\end{align}
optimizations in $|\set{X}_{1}| \cdot N_A-1$ dimensions.
In contrast, \eqref{eq:rc-fb-0} and \eqref{eq:EG07-Thm2}
require optimizing $P_{X_{1}{\mathbf A}_{2}}$
for $|\set{X}_{2}|^{|\set{A}_{2}|\cdot |\set{Y}_{2}|}$ functions
$f(\cdot): \set{A}_{2}\times\set{Y}_{2}\rightarrow\set{X}_{2}$ where $|\set{A}_{2}|$ is
at most
\begin{align}
N_V=|\set{X}_{1}| \cdot |\set{X}_{2}|+1.
\end{align}
We thus have at most $|\set{X}_{2}|^{N_V\cdot |\set{Y}_{2}|}$
optimizations in $|\set{X}_{1}| \cdot N_V-1$ dimensions. But we have $N_A \le N_V$ and
\begin{align}
\binom{|\set{X}_{2}|^{|\set{Y}_{2}|}}{N_A}
\le |\set{X}_{2}|^{N_A \cdot |\set{Y}_{2}|} \le |\set{X}_{2}|^{N_V \cdot |\set{Y}_{2}|}
\end{align}
so the optimization of \eqref{eq:rc-fb-0}-\eqref{eq:rc-fb-0-cardinality}
is generally simpler than the optimization of \eqref{eq:rc-fb-0} and
\eqref{eq:EG07-Thm2}. This discussion shows that one
may as well consider code functions directly rather than introducing auxiliary
random variables and auxiliary functions.
\end{remark}
\begin{example}
Suppose that $|\set{X}_{1}|=|\set{X}_{2}|=2$ and $|\set{Y}_{2}|=4$.
Then \eqref{eq:rc-fb-0-cardinality} states that at most $5$ code functions (here code trees) out of
16 need have positive probability. Our search is thus over $\binom{16}{5}=4368$
combinations of code trees. In comparison, \cite[Thm. 2]{ElGamal:07}
requires a search over $2^{20}\approx10^{6}$ mappings $f(\cdot)$. 
\end{example}

\subsection{Relay Networks with Delays}
\label{subsec:relay-networks-delays}
Relay networks with delays~\cite{ElGamal:07} have the simplifying feature that every
node has at most one channel input and output in each block. Furthermore, there is exactly
one network message that originates at a designated source node $k=1$ and that is
destined for a designated node $k=K$. Nodes $1$ and $K$ have no channel
outputs and inputs, respectively, i.e., we effectively have $Y_{1,i}=X_{K,i}=0$ for all $i$.

A cut bound for such networks was developed in~\cite[Thm.~4]{ElGamal:07}
that is almost the same as Theorem~\ref{thm:cut-set-bound}. The difference between
the bounds is similar to the difference described in Remark~\ref{remark:bound}
above, i.e.,~\cite[Thm.~4]{ElGamal:07} uses auxiliary variables for the code functions
(in this case Shannon strategies) and specifies cardinality bounds on these variables.
Theorem~\ref{thm:cut-set-bound} instead uses the code functions directly, 
and these functions have finite cardinality if the channel input and output alphabets are finite
(see Remark~\ref{remark:auxRV}). One may develop improved cardinality bounds as
in~\cite[Thm.~4]{ElGamal:07} that are useful if the channel input or output alphabets are
continuous.

\subsection{Causal Relay Networks and Generalized Networks}
\label{subsec:generalized}
Causal relay networks~\cite{Baik11} and generalized networks~\cite{FongIT:12}
are NiBMs that extend relay networks with delays by considering more than
one unicast session. We describe these networks by using an example
with $K=5$ nodes whose FDG for one block is shown in Fig.~\ref{fig:causal-rn}.
Nodes 1 and 2 can encode by using only received symbols from {\em past}
NiBM blocks and they are called {\em strictly causal} relays.
Nodes 3, 4, and 5 can encode by using received symbols from past {\em and current}
NiBM blocks and they are called {\em causal} relays. The block length is $L=3$.

\begin{figure}[t!]
  \centerline{\includegraphics[scale=0.5]{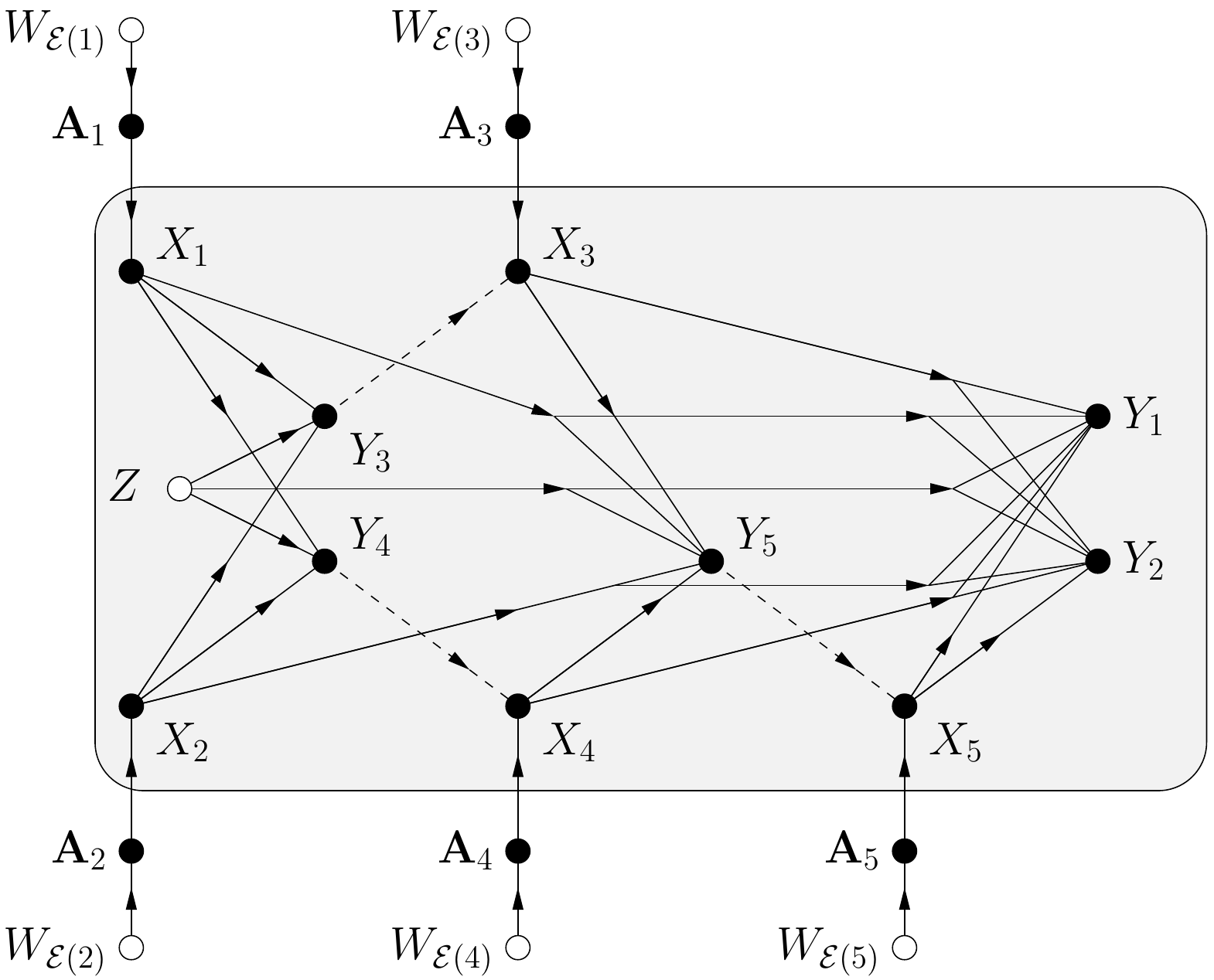}}
  \caption{FDG for a causal relay network with $K=5$ nodes
   and $n=3$ channel uses. The network is a NiBM with block length $L=3$.}
  \label{fig:causal-rn}
\end{figure}

In the language of~\cite{Baik11}, the strictly causal relays are in the set $\set{N}_1=\{1,2\}$
and the causal relays are in $\set{N}_0=\{3,4,5\}$.
In the language of~\cite[Defn.~1]{FongIT:12}, we have two $3$-partitions of
$\set{K}=\{1,2,3,4,5\}$, namely the input and output partitions $\boldsymbol{\set{S}}$
$\boldsymbol{\set{G}}$ where
\begin{align}
   & \boldsymbol{\set{S}} = \{ \set{S}_1=\{1,2\}, \set{S}_2=\{3,4\}, \set{S}_3=\{5\}  \} \nonumber \\
   & \boldsymbol{\set{G}} = \{ \set{G}_1=\{3,4\}, \set{G}_2=\{5\}, \set{G}_3=\{1,2\}  \} .
\end{align}
We do not consider this notation further and focus on arguing that
Theorem~\ref{thm:cut-set-bound} improves on the main results of~\cite{Baik11,FongIT:12}.

Consider first~\cite[Thm.~1]{Baik11} and~\cite[Thm.~1]{FongIT:12}.
These bounds are the same as Theorem~\ref{thm:cut-set-bound} except that 
the right-hand side of \eqref{eq:weak} (normalized by $L$) replaces the right-hand side of
\eqref{eq:rate-bound}. We conclude that Theorem~\ref{thm:cut-set-bound} is at least
as good as~\cite[Thm.~1]{Baik11} and \cite[Thm.~1]{FongIT:12}.
Moreover, Example~\ref{example:state} shows that Theorem~\ref{thm:cut-set-bound}
can strictly improve these bounds (see also Example~\ref{ex:weak}).

Consider next~\cite[Thm.~2]{Baik11}. 
We illustrate how the bound works by using the cut $\set{S}=\{1,3\}$ in the
network of Fig.~\ref{fig:causal-rn}.
Theorem~\ref{thm:cut-set-bound} and a series of further steps gives
\begin{align}
  3R_{\set{M}(\set{S})}
  & \overset{(a)}{\le} I\left( X_1 {\mathbf A_3} ; Y_2 Y_4 Y_5 | X_2 {\mathbf A_4} {\mathbf A_5} \right)
  \nonumber \\
  & \overset{(b)}{=} I\left( X_1 ; Y_4  | X_2 {\mathbf A_4} {\mathbf A_5} \right) \nonumber \\ & \quad
  + I\left( X_1 {\mathbf A_3} ; Y_5 | X_2 X_4 Y_4 {\mathbf A_4} {\mathbf A_5} \right) \nonumber \\
  & \quad + I\left( X_1 {\mathbf A_3} ; Y_2 | X_2 X_4 X_5 Y_4 Y_5 {\mathbf A_4} {\mathbf A_5} \right)
  \nonumber \\
  & \overset{(c)}{\le} I\left( X_1 ; Y_4  | X_2 {\mathbf A_4} {\mathbf A_5} \right) \nonumber \\ & \quad
  + I\left( X_1 X_3 Y_3 ; Y_5 | X_2 X_4 Y_4 {\mathbf A_4} {\mathbf A_5} \right) \nonumber \\
  & \quad + I\left( X_1 {\mathbf A_3}; Y_2 | X_2 X_4 X_5 Y_4 Y_5 {\mathbf A_4} {\mathbf A_5} \right)
    \nonumber \\
  & \overset{(d)}{\le} I\left( X_1 ; Y_4  | X_2 \right) \nonumber \\ & \quad
  + I\left( X_1 X_3 Y_3 ; Y_5 | X_2 X_4 Y_4 \right) \nonumber \\
  & \quad + I\left( X_1 X_3 Y_3; Y_2 | X_2 X_4 X_5 Y_4 Y_5 \right)
  \label{eq:causal-relay-bounds}
\end{align}
where $(a)$ is simply \eqref{eq:rate-bound} and $(b)$ follows by using the chain rule for
mutual information and the Markovity in the channel.
Step $(c)$ follows because we have added $Y_3$ to the second mutual information expression
and by using the Markovity in the channel. The result is the bound of~\cite[Thm.~2]{Baik11} when
the causal relays do not have messages.
Step $(d)$ follows similarly and is the bound of~\cite[Thm.~1]{Baik11} and~\cite[Thm.~1]{FongIT:12}.

The above example extends to any causal relay network and any cut
(see Appendix~\ref{app:causal-relay-networks}).
In other words, the bound of~\cite[Thm.~2]{Baik11} improves
on the bounds of~\cite[Thm.~1]{Baik11} and~\cite[Thm.~1]{FongIT:12},
but all three bounds are implied by Theorem~\ref{thm:cut-set-bound}.
We show in Example~\ref{example:causal} below that
if the causal relays have no messages then Theorem~\ref{thm:cut-set-bound}
can be strictly better than~\cite[Thm.~2]{Baik11} due to inequality $(c)$.
Furthermore, the auxiliary random variables $U_k$ in~\cite[Thm.~2]{Baik11} are
not specified to be code functions. The optimization is thus more complex than by using
Theorem~\ref{thm:cut-set-bound} in general (see Remark~\ref{remark:bound}).
\begin{example} \label{example:causal}
Consider Fig.~\ref{fig:causal-rn} with $\set{X}_k=\set{Y}_k=\{0\}$ for $k=2,4$, i.e.,
nodes 2 and 4 are removed from the problem.
Consider $Y_3=[X_1,Z]$ where $\set{X}_1=\{0,1\}$ and
$P_Z(0)=P_Z(1)=1/2$, and $Y_5=Z$. Suppose there is only one
message with rate $R_{15}$ at node 1 destined for node 5
(so the causal relays at nodes 3 and 5 have no messages).
We effectively have a RC with no delay and the capacity is zero because
$X_1{\mathbf A}_3$ has no influence on $Y_5$.
For instance, the cut-set bound \eqref{eq:rate-bound} with $\set{S}=\{1,3\}$ gives
$3R_{15}\le I(X_1 {\mathbf A}_3 ; Y_5 | {\mathbf A}_5)=0$.

Next, consider the cut-set bound of~\cite[Thm.~2]{Baik11}.
There are two cuts to consider without nodes 2 and 4.
The cut $\set{S}=\{1,3\}$ gives (see \eqref{eq:causal-relay-bounds} after step $(c)$)
\begin{align}
   3R_{15} \le I(X_1 X_3 Y_3 ; Y_5  | {\mathbf A}_5 ) = 1
\end{align}
and the cut $\set{S}=\{1\}$ gives
\begin{align}
   3R_{15} \le I(X_1 ; Y_3 Y_5 | {\mathbf A}_3 {\mathbf A}_5 ) = H(X_1 | {\mathbf A}_3 {\mathbf A}_5 ).
\end{align}
But we have $H(X_1 | {\mathbf A}_3 {\mathbf A}_5)=1$ by choosing $X_1$
independent of ${\mathbf A}_3 {\mathbf A}_5$ and $P_{X_1}(0)=P_{X_1}(1)=1/2$.
Thus, the cut-set bound of~\cite[Thm.~2]{Baik11} is loose while Theorem~\ref{thm:cut-set-bound} is tight.
\end{example}
\begin{example}
Consider the generalized network called a ``BSC with correlated feedback" in~\cite[Sec.~VI]{FongIT:12}.
This network is a two-way channel with iBM and block length $L=2$ and with binary inputs and outputs
\begin{align*}
   & Y_{2,1} = X_{1,1} \oplus Z \\
   & Y_{1,2} = X_{2,2} \oplus Y_{2,1}
\end{align*}
where $P_Z(1)=1-P_Z(0)=\epsilon$.
The rate pair $(R_1,R_2)=(1-H_2(\epsilon),1)/2$ is achievable by
choosing $X_{1,1}$ as uniform over $\{0,1\}$ and $X_{2,2}=X_{2,2}' \oplus Y_{2,1}$
where $X_{2,2}'$ is independent of $Y_{2,1}$ and uniform over $\{0,1\}$.
For the converse, the cut-set bound of Theorem~\ref{thm:cut-set-bound} is
\begin{align} \label{eq:gn}
 \bigcup_{P_{X_{1,1}{\mathbf A}_{2,2}}} \left\{ (R_1,R_2):
 \begin{array}{l}
  0 \le R_1 \le I(X_{1,1} ; Y_{2,1} | {\mathbf A}_{2,2} )/2 \\
  0 \le R_2 \le I({\mathbf A}_{2,2} ; Y_{1,2} | X_{1,1} )/2 \\
  \end{array} \right\}
\end{align}
and we have $I(X_{1,1} ; Y_{2,1} | {\mathbf A}_{2,2} ) \le 1-H_2(\epsilon)$ with equality
if $X_{1,1}$ is uniform and independent of ${\mathbf A}_{2,2}$.
We further have $I( {\mathbf A}_{2,2} ; Y_{1,2} | X_{1,1})\le 1$ since $Y_{1,2}$ is binary.
This shows that Theorem~\ref{thm:cut-set-bound} is tight.

Finally, we translate the capacity-achieving strategy into a code tree distribution.
We label the branch-pairs of our tree ${\mathbf A}_{2,2}$ as $b_0b_1$ by which we mean
that $X_{2,2}=b_0$ if $Y_{2,1}=0$ and $X_{2,2}=b_1$ if $Y_{2,1}=1$.  We choose
${\mathbf A}_{2,2}$ independent of $X_{1,1}$ and
\begin{align*}
  & P_{{\mathbf A}_{2,2}}(00)=P_{{\mathbf A}_{2,2}}(11)=0 \\
  & P_{{\mathbf A}_{2,2}}(01)=P_{{\mathbf A}_{2,2}}(10)=1/2
\end{align*}
and compute $I( {\mathbf A}_{2,2} ; Y_{1,2} | X_{1,1})=1$, as desired.
\end{example}

\subsection{Quantize-Forward Network Coding}
\label{subsec:NC}
The final channels we consider are relay networks with iBM. Suppose node $1$ multicasts
a message of rate $R$ to sink nodes in the set $\set{T}$. The quantize-map-forward (QMF) and
noisy network coding (NNC) strategies in~\cite{Avestimehr11,YA:11,Lim11} generalize to NiBMs and
we call the resulting strategies {\em quantize-forward} (QF) network coding.
QF network coding achieves $R$ satisfying
\begin{align} \label{eq:QF-NC-rate}
 LR \le \min_{k \in \set{S}^c \cap \set{T}} \;
 & I( {\mathbf A}_{\set{S}}^L ; \hat{Y}_{\set{S}^c}^L Y_k | {\mathbf A}_{\set{S}^c}^L T ) \nonumber \\
 & - I( Y_{\set{S}}^L ; \hat{Y}_{\set{S}}^L | {\mathbf A}_{\set{K}}^L \hat{Y}_{\set{S}^c}^L T )
\end{align}
for all $\set{S}\subset\set{K}$ with $1\in\set{S}$ and $\set{S}^c \cap \set{T} \ne \emptyset$.
The ${\mathbf A}_k^L$, $k=1,2,\ldots,K$, are
independent and $\hat{Y}_k^L$ is a noisy function of ${\mathbf A}_k^L$
and $Y_k^L$ for all $k$.

\begin{remark}
A simple lower bound on the first mutual information expression in
\eqref{eq:QF-NC-rate} is
\begin{align}
   I( {\mathbf A}_{\set{S}}^L ; \hat{Y}_{\set{S}^c}^L Y_k | {\mathbf A}_{\set{S}^c}^L T )
   \ge I( {\mathbf A}_{\set{S}}^L ; \hat{Y}_{\set{S}^c}^L | {\mathbf A}_{\set{S}^c}^L T ).
   \label{eq:QF-NC-rate-lb}
\end{align}
We use the right-hand side of \eqref{eq:QF-NC-rate-lb} below because it better matches
\eqref{eq:rate-bound} with $\hat{Y}_{\set{S}^c}^L$ replacing $Y_{\set{S}^c}^L$. 
\end{remark}
\begin{example}
We extend results of~\cite{Avestimehr11,YA:11,Lim11}.
If the network is deterministic then ${\mathbf A}_{\set{K}}^L$ determines
$X_{\set{K}}^L Y_{\set{K}}^L$. We thus have
\begin{align}
   I( Y_{\set{S}}^L ; \hat{Y}_{\set{S}}^L | {\mathbf A}_{\set{K}}^L \hat{Y}_{\set{S}^c}^L T ) = 0
\end{align}
and can choose $\hat{Y}_k^L=Y_k^L$ to achieve the cut-set bound but
evaluated with independent code functions only. As a result, we obtain the multicast capacity
of networks of deterministic point-to-point channels with iBM, for instance. However, 
QF network coding does not give the capacity region for all deterministic networks because dependent
code functions may increase rates.
\end{example}

\subsection{QF Network Coding for Gaussian Networks}
\label{subsec:NC-Gauss}
Consider the channel \eqref{eq:linear-channel} with additive Gaussian noise (AGN), i.e.,
the $\ul{Z}_k$  are Gaussian noise vectors and where $\ul{Z}_{\set{K}}$ has a positive definite covariance matrix.
For simplicity, we assume that the $\ul{Z}_1, \ul{Z}_2,\ldots,\ul{Z}_K$ are mutually independent.

Suppose again that node $1$ multicasts a message of rate $R$ to sink nodes in $\set{T}$.
Let $\set{S}$ be a cut, i.e., $1\in\set{S}$ and $\set{S}^c \cap \set{T} \ne \emptyset$.
We use the notation
\begin{align} \label{eq:linear-channel-2}
 \ul{Y}_{\set{S}^C} = {\mathbf G}_{\set{S}^c\set{S}} \ul{X}_{\set{S}} + {\mathbf G}_{\set{S}^c\set{S}^c} \ul{X}_{\set{S}^c} +  \ul{Z}_{\set{S}^c}
\end{align}
for the $|\set{S}^c|$ equations \eqref{eq:linear-channel}
with $k \in \set{S}^c$, where ${\mathbf G}_{\set{U}\set{V}}$ is a $|\set{U}| L \times |\set{V}| L$ matrix
with block-entries ${\mathbf G}_{kj}$, $k\in\set{U}, j\in\set{V}$. Recall that the ${\mathbf G}_{kj}$
are $L \times L$ lower-triangular matrices.

We begin with the upper bound \eqref{eq:weak2} which we write as
\begin{align}
 & h({\mathbf G}_{\set{S}^c\set{S}} \ul{X}_{\set{S}}  +  \ul{Z}_{\set{S}^c} \| \ul{X}_{\set{S}^c} ) - h( \ul{Z}_{\set{S}^c} ) \nonumber \\
 & \le h({\mathbf G}_{\set{S}^c\set{S}} \ul{X}_{\set{S}}  +  \ul{Z}_{\set{S}^c} ) - h( \ul{Z}_{\set{S}^c} ) \nonumber \\
 & \overset{(a)}{\le}
     \frac{1}{2} \log \frac{\left| {\mathbf Q}_{\ul{Z}_{\set{S}^c}} + {\mathbf G}_{\set{S}^c\set{S}} \, {\mathbf Q}_{\ul{X}_{\set{S}}} \,
     {\mathbf G}_{\set{S}^c\set{S}}^T \right|}{\left|  {\mathbf Q}_{\ul{Z}_{\set{S}^c}}  \right|}
 \label{eq:rn-bound1}
\end{align}
where $(a)$ follows by a classic maximum entropy theorem. The (positive definite) noise covariance matrix
has a Cholesky decomposition
${\mathbf Q}_{\ul{Z}_{\set{S}^c}} = {\mathbf S}_{\ul{Z}_{\set{S}^c}} {\mathbf S}_{\ul{Z}_{\set{S}^c}}^T $
where ${\mathbf S}_{\ul{Z}_{\set{S}^c}}$ is lower triangular and invertible. We can thus rewrite \eqref{eq:rn-bound1} as
\begin{align}
 & I(X_{\set{S}}^L \rightarrow Y_{\set{S}^c}^L \| X_{\set{S}^c}^L )
    \le \frac{1}{2} \log \left| {\mathbf I}_{\set{S}^c} + \tilde{{\mathbf G}}_{\set{S}^c\set{S}} \, {\mathbf Q}_{\ul{X}_{\set{S}}} \,
     \tilde{{\mathbf G}}_{\set{S}^c\set{S}}^T \right|
 \label{eq:rn-bound2}
\end{align}
where ${\mathbf I}_{\set{U}}$ is the $|\set{U}| L \times |\set{U}| L$ identity matrix
and $\tilde{{\mathbf G}}_{\set{S}^c\set{S}}={\mathbf S}_{\ul{Z}_{\set{S}^c}}^{-1} {\mathbf G}_{\set{S}^c\set{S}}$.

For achievability, we choose $T$ to be a constant and the code functions (effectively) as codewords
\begin{align} \label{eq:lower-distrib}
 & {\mathbf A}_k^L(\cdot) = X_k^L, \quad k=1,2,\ldots,K 
\end{align}
where $X_k^L$ is Gaussian. We further choose
\begin{align}
 & \hat{Y}_k^L = Y_k^L + \hat{Z}_k^L, \quad k=1,2,\ldots,K
\end{align}
where $\hat{Z}_{\set{K}}^L$ is independent of $X_{\set{K}}^L Y_{\set{K}}^L$ and has the
same statistics as $Z_{\set{K}}^L$. Consider the right-hand side of \eqref{eq:QF-NC-rate-lb} with
codewords rather than code functions. We have
\begin{align}
 & I( X_{\set{S}}^L ; \hat{Y}_{\set{S}^c}^L | X_{\set{S}^c}^L ) \nonumber \\
 & \overset{(a)}{=} h({\mathbf G}_{\set{S}^c\set{S}} \ul{X}_{\set{S}}  +  \ul{Z}_{\set{S}^c} +  \hat{\ul{Z}}_{\set{S}^c})
     - h( \ul{Z}_{\set{S}^c}  +  \hat{\ul{Z}}_{\set{S}^c}) \nonumber \\
 & \overset{(b)}{=} \frac{1}{2} \log \frac{\left| 2 {\mathbf Q}_{\ul{Z}_{\set{S}^c}} + {\mathbf G}_{\set{S}^c\set{S}} \,
     {\mathbf Q}_{\ul{X}_{\set{S}}} \, {\mathbf G}_{\set{S}^c\set{S}}^T \right|}
     {\left|  2 {\mathbf Q}_{\ul{Z}_{\set{S}^c}}  \right|} \nonumber \\
  & = \frac{1}{2} \log \left| {\mathbf I}_{\set{S}^c} + \frac{1}{2} \tilde{{\mathbf G}}_{\set{S}^c\set{S}} \, {\mathbf Q}_{\ul{X}_{\set{S}}} \,
      \tilde{{\mathbf G}}_{\set{S}^c\set{S}}^T \right| \nonumber \\
  & \overset{(c)}{\ge} \frac{1}{2} \log \left| {\mathbf I}_{\set{S}^c} + \tilde{{\mathbf G}}_{\set{S}^c\set{S}} \, {\mathbf Q}_{\ul{X}_{\set{S}}} \,
      \tilde{{\mathbf G}}_{\set{S}^c\set{S}}^T \right| -\frac{|\set{S}^c| L}{2}    
 \label{eq:rn-bound3}
\end{align}
where $(a)$ is because the $X_k^L$ are independent, $(b)$ is because the $X_k^L$ are Gaussian,
and $(c)$ follows by using
$|{\mathbf A}+{\mathbf B}/2| \ge |({\mathbf A}+{\mathbf B})/2| = |{\mathbf A}+{\mathbf B}|/2^b$
for $b \times b$ positive definite matrices ${\mathbf A}$ and ${\mathbf B}$.
We also have
\begin{align}
 I( Y_{\set{S}}^L ; \hat{Y}_{\set{S}}^L | X_{\set{K}}^L \hat{Y}_{\set{S}^c}^L )
 & = I( Z_{\set{S}}^L ; Z_{\set{S}}^L + \hat{Z}_{\set{S}}^L | X_{\set{K}}^L \hat{Z}_{\set{S}^c}^L ) \nonumber \\
 & = I( Z_{\set{S}}^L ; Z_{\set{S}}^L + \hat{Z}_{\set{S}}^L ) \nonumber \\ 
 & = |\set{S}|L/2 \label{eq:rn-bound4}
\end{align}
where the last step is because $\hat{Z}_{\set{S}}^L$ has the same
statistics as $Z_{\set{S}}^L$.
Combining \eqref{eq:rn-bound3} and \eqref{eq:rn-bound4} we find that $R$ satisfying
\begin{align}
    LR \le \frac{1}{2} \log \left| {\mathbf I}_{\set{S}^c} + \tilde{{\mathbf G}}_{\set{S}^c\set{S}} \, {\mathbf Q}_{\ul{X}_{\set{S}}} \,
      \tilde{{\mathbf G}}_{\set{S}^c\set{S}}^T \right| -\frac{K L}{2}    
 \label{eq:rn-bound5}
\end{align}
for all $\set{S}\subset\set{K}$ with $1\in\set{S}$ and $\set{S}^c \cap \set{T} \ne \emptyset$
are achievable.

It remains to study the first expression on the right-hand side of \eqref{eq:rn-bound5}, both without and
with independent $X_k^L$. Suppose that $\tilde{{\mathbf G}}_{\set{S}^c\set{S}}$ has the singular value
decomposition ${\mathbf U}^T {\mathbf \Sigma} {\mathbf V}$ so that this expression is
\begin{align}
       \frac{1}{2} \log \left| {\mathbf I}_{\set{S}^c} + {\mathbf \Sigma} {\mathbf V} \, {\mathbf Q}_{\ul{X}_{\set{S}}} \,
       {\mathbf V}^T {\mathbf \Sigma}^T \right| .
 \label{eq:rn-bound6}
\end{align}
Suppose there are $K$ power constraints $\sum_{i=1}^n {\rm E}[X_{k,i}^2]/n \le P$, $k=1,2,\ldots,K$,
i.e., we have {\em symmetric} power constraints. Optimizing over
${\mathbf Q}_{\ul{X}_{\set{S}}}$ we obtain $\min(|\set{S}|,|\set{S}^c|) \cdot L$ parallel channels on
which we can put at most power $|\set{S}| P$. We thus have the capacity upper bound
\begin{align}
   LR & \le \sum_j \frac{1}{2} \log\left( 1 + s_j^2 |\set{S}| P \right)
 \label{eq:rn-bound7}
\end{align}
where the sum is over the parallel channels and the $s_j$ are the singular values.

For a lower bound we simplify \eqref{eq:lower-distrib} even further and choose
${\mathbf Q}_{\ul{X}_k}=(P/L) \cdot {\mathbf I}_{\{k\}}$.
The expression \eqref{eq:rn-bound6} becomes
\begin{align}
   & \sum_{s_j} \frac{1}{2} \log\left( 1 + s_j^2 (P/L) \right) \nonumber \\
   & \ge  \left[ \sum_{s_j} \frac{1}{2} \log\left( 1 + s_j^2 |\set{S}| P \right)\right] - \frac{|\set{S}| L}{2}\log\left( |\set{S}| L \right).
 \label{eq:rn-bound8}
\end{align}
We thus have the following theorem that implies that QF network coding approaches capacity
at high signal-to-noise ratio. This extends results in~\cite{Avestimehr11,YA:11,Lim11} to NiBMs.
\begin{theorem} \label{thm:QF-NC}
QF network coding for scalar, linear, AGN channels, symmetric power constraints, and a multicast session
achieves capacity to within
\begin{align} \label{eq:QF-NC-rate-bound}
 K(1+\log(KL))/2 \text{ bits}.
\end{align}
\end{theorem}
One may derive better results than \eqref{eq:QF-NC-rate-bound} by using the
approach in~\cite{Lim11}, for example. Extensions to asymmetric power
constraints and multiple multicast sessions are clearly possible.

\setcounter{section}{0}
\renewcommand{\thesection}{\Alph{section}}
\renewcommand{\thesubsection}{\arabic{subsection}}
\renewcommand{\appendix}[1]{%
  \refstepcounter{section}%
  \par\begin{center}%
    \begin{sc}%
      Appendix \thesection\par%
      #1%
    \end{sc}%
  \end{center}\nobreak%
}

\bigskip
\appendix{Proof of Cut-Set Bound}
\label{app:proof}
The bound follows from classic steps and the factorizations
\eqref{eq:joint-pdf} and \eqref{eq:channel-pdf}. There is one new subtlety, however, namely how to
define the random code functions that appear in \eqref{eq:rate-bound}.
Fano's inequality states that for $P_e\rightarrow0$ we have
\begin{align}
n R_{\set{M}(\set{S})} & \le I( W_{\set{M}(\set{S})} ; \{ \hat{W}_{\set{M}(\set{S})}^{(\ell)}: \ell\in\set{S}^c \} ) \nonumber \\
& \overset{(a)}{\le} I( W_{\set{E}(\set{S})} ; Y_{\set{S}^c}^n W_{\set{E}(\set{S}^c)} ) \nonumber \\
& \overset{(b)}{=} I( W_{\set{E}(\set{S})} {\mathbf A}_{\set{S}}^n ; Y_{\set{S}^c}^n | W_{\set{E}(\set{S}^c)} {\mathbf A}_{\set{S}^c}^n ) \nonumber \\
& \overset{(c)}{=} I( {\mathbf A}_{\set{S}}^n ; Y_{\set{S}^c}^n | {\mathbf A}_{\set{S}^c}^n )
\label{eq:bound1}
\end{align}
where $(a)$ follows because $\hat{W}_{\set{M}(\set{S})}$ is a subset of $\hat{W}_{\set{E}(\set{S})}$ and because
 $\{ \hat{W}_{\set{M}(\set{S})}^{(\ell)}: \ell\in\set{S}^c \}$ is a function of $Y_{\set{S}^c}^n$ and $W_{\set{E}(\set{S}^c)}$;
$(b)$ follows because the messages are independent and ${\mathbf A}_k^n$ is a function of the messages at node $k$;
and $(c)$ follows because
\begin{align}
  W_{\set{E}(\set{S})} - {\mathbf A}_{\set{S}}^L - Y_{\set{S}'}^L
\end{align}
forms a Markov chain
for any $\set{S}$ and $\set{S}'$. Recall that $n=mL$ for some integer $m$. We may thus write
\begin{align}
I( {\mathbf A}_{\set{S}}^n ; Y_{\set{S}^c}^n | {\mathbf A}_{\set{S}^c}^n )
& \overset{(a)}{=} \sum_{i=1}^m I( {\mathbf A}_{\set{S}}^n ; Y_{\set{S}^c,i}^{L} | {\mathbf A}_{\set{S}^c}^n Y_{\set{S}^c}^{(i-1)L} ) \nonumber \\
& \overset{(b)}{=} \sum_{i=1}^m I( {\mathbf A}_{\set{S}}^{iL} ; Y_{\set{S}^c,i}^{L} | {\mathbf A}_{\set{S}^c}^{iL} Y_{\set{S}^c}^{(i-1)L} ) \nonumber \\
& \le \sum_{i=1}^m I( {\mathbf A}_{\set{S}}^{iL} Y_{\set{S}}^{(i-1)L} ; Y_{\set{S}^c,i}^{L} | {\mathbf A}_{\set{S}^c}^{iL} Y_{\set{S}^c}^{(i-1)L} )
\label{eq:bound2}
\end{align}
where $(a)$ follows by choosing $Y_{k,i}^L$ to be the channel output of node $k$ from time $(i-1)L+1$ to time $iL$,
and where $(b)$ follows by Markovity.

Now let $\bar{\mathbf A}_{k,i}^L$ be the string of functions
${\mathbf A}_{k,j}(\cdot,Y_k^{(i-1)L})$, $j=(i-1)L+1,(i-1)L+2,\ldots,iL$. We then have
\begin{align}
& I( {\mathbf A}_{\set{S}}^{iL} Y_{\set{S}}^{(i-1)L} ; Y_{\set{S}^c,i}^{L} | {\mathbf A}_{\set{S}^c}^{iL} Y_{\set{S}^c}^{(i-1)L} ) \nonumber \\
& \overset{(a)}{=}H( Y_{\set{S}^c,i}^{L} | \bar{\mathbf A}_{\set{S}^c,i}^{L} {\mathbf A}_{\set{S}^c}^{iL} Y_{\set{S}^c}^{(i-1)L} )
   - H(Y_{\set{S}^c,i}^{L} |  \bar{\mathbf A}_{\set{K},i}^{L} {\mathbf A}_{\set{K}}^{iL} Y_{\set{K}}^{(i-1)L} ) \nonumber \\
& \le H( Y_{\set{S}^c,i}^{L} | \bar{\mathbf A}_{\set{S}^c,i}^{L} )
    - H(Y_{\set{S}^c,i}^{L} | \bar{\mathbf A}_{\set{K},i}^{L} {\mathbf A}_{\set{K}}^{iL} Y_{\set{K}}^{(i-1)L} ) \nonumber \\
& \overset{(b)}{=} H( Y_{\set{S}^c,i}^{L} | \bar{\mathbf A}_{\set{S}^c,i}^{L} ) - H(Y_{\set{S}^c,i}^{L} | \bar{\mathbf A}_{\set{K},i}^{L} ) \nonumber \\
& =  I( \bar{\mathbf A}_{\set{S},i}^{L} ; Y_{\set{S}^c,i}^{L} | \bar{\mathbf A}_{\set{S}^c,i}^{L} )
\label{eq:bound3}
\end{align}
where $(a)$ follows because $\bar{\mathbf A}_{k,i}^L$ is a function of ${\mathbf A}_k^{iL} Y_k^{(i-1)L}$
and $(b)$ follows because
\begin{align}
  {\mathbf A}_{\set{K}}^{iL} Y_{\set{K}}^{(i-1)L} - \bar{\mathbf A}_{\set{K},i}^{L} - Y_{\set{S}^c,i}^L
\end{align}
forms a Markov chain (this step permits $L$-letterization).

The remaining steps follow because the
$\bar{\mathbf A}_{\set{K}}^L$-to-$Y_{\set{K}}^L$ channel does not depend
on the block index $i$. More precisely, we have
\begin{align}
& P(y_{\set{K},i}^L | \bar{\mathbf a}_{\set{K},i}^L )
   = P_{Y_{\set{K}}^L |  {\mathbf A}_{\set{K}}^L}(y_{\set{K},i}^L | \bar{\mathbf a}_{\set{K},i}^L ) \nonumber \\
&  = \left[ \prod_{k=1}^K  1(x_{k}^L \| \bar{\mathbf a}_{k,i}^L,0y_{k,i}^{L-1}) \right]
    P_{Y_{\set{K}}^L \|  X_{\set{K}}^L}(y_{\set{K},i}^L \| x_{\set{K}}^L )
     \label{eq:cut-set-pdf-2a}
\end{align}
where $P_{Y_{\set{K}}^L |  {\mathbf A}_{\set{K}}^L}$ refers to the first $L$ channel uses.
Inserting \eqref{eq:bound3} into \eqref{eq:bound2}, we have
\begin{align}
I( {\mathbf A}_{\set{S}}^n ; Y_{\set{S}^c}^n | {\mathbf A}_{\set{S}^c}^n )
& \le \sum_{i=1}^m I( \bar{\mathbf A}_{\set{S},i}^L ; Y_{\set{S}^c,i}^L | \bar{\mathbf A}_{\set{S}^c,i}^L ) \nonumber \\
& = m I( \bar{\mathbf A}_{\set{S},T}^L ; Y_{\set{S}^c,T}^L | \bar{\mathbf A}_{\set{S}^c,T}^L T ) \nonumber \\
& \overset{(a)}{\le} m I( \bar{\mathbf A}_{\set{S},T}^L ; Y_{\set{S}^c,T}^L | \bar{\mathbf A}_{\set{S}^c,T}^L )
\label{eq:bound4}
\end{align}
where $T$ takes on the value $i$, $i=1,2,\ldots,m$, with probability $1/m$, and
where $(a)$ follows because
\begin{align}
  T - \bar{\mathbf A}_{\set{K},T}^L - Y_{\set{K},T}^L
\end{align}
forms a Markov chain.
Inserting \eqref{eq:bound4} into \eqref{eq:bound1}, we have
\begin{align}
L \cdot R_{\set{M}(\set{S})} & \le I( \bar{\mathbf A}_{\set{S},T}^L ; Y_{\set{S}^c,T}^L | \bar{\mathbf A}_{\set{S}^c,T}^L )
\label{eq:bound-final}
\end{align}
where the joint distribution of the random variables factors as
\begin{align} \label{eq:cut-set-pdf-2b}
& P(\bar{\mathbf a}_{\set{K},T}^L) P_{Y_{\set{K}}^L |  {\mathbf A}_{\set{K}}^L}(y_{\set{K},T}^L | \bar{\mathbf a}_{\set{K},T}^L)
\end{align}
and where the second term in \eqref{eq:cut-set-pdf-2b} is computed using \eqref{eq:cut-set-pdf-2a}
(this step permits the factorization \eqref{eq:cut-set-pdf}).

\begin{remark}
If $n\ne mL$ then we may as well consider $n=(m-1)L+L'$ where $0<L'<L$.
The sum in \eqref{eq:bound4} changes and has as its $m$th term
\begin{align}
I( \bar{\mathbf A}_{\set{S},m}^{L'} ; Y_{\set{S}^c,m}^{L'} | \bar{\mathbf A}_{\set{S}^c,m}^{L'} )
\label{eq:bound-extra}
\end{align}
where the code functions have depth $L'$. The term \eqref{eq:bound-extra} could be larger than the
right-hand side of \eqref{eq:bound-final}. However, if
$m$ is large then the capacity is effectively limited by \eqref{eq:bound-final}.
\end{remark}
\begin{remark} \label{rmk:cost-constraint}
Consider the $j$th cost constraint in \eqref{eq:cost-constraints}. We may rewrite \eqref{eq:cost-constraints} as
\begin{align}
  & \frac{1}{L} \sum_{\ell=1}^L \frac{1}{m} \sum_{i=1}^m \E{s_j\left( X_{\set{K},(m-1)L+\ell}, Y_{\set{K},(m-1)L+\ell} \right)} \nonumber \\
  & = \frac{1}{L} \sum_{\ell=1}^L \E{s_j\left( X_{\set{K},(T-1)L+\ell}, Y_{\set{K},(T-1)L+\ell} \right)} 
      \le S_j \label{eq:cost-constraints3}
\end{align}
and the inequality in \eqref{eq:cost-constraints3} is the $j$th inequality in \eqref{eq:cost-constraints2}.
\end{remark}

\appendix{Cardinality Bounds For Point-to-Point Channels}
\label{app:cardinality-ptp}
Consider a point-to-point channel with NiBM. We write
\begin{align}
   & P(y^L) = \sum_{{\mathbf a}^L} P({\mathbf a}^L) P(y^L | {\mathbf a}^L) \label{eq:app-P} \\
   & H(Y^L | {\mathbf A}^L) = \sum_{{\mathbf a}^L} P({\mathbf a}^L) H(Y^L | {\mathbf A}^L={\mathbf a}^L) \label{eq:app-H}
\end{align}
where $P(y^L | {\mathbf a}^L)$ and $H(Y^L | {\mathbf A}^L={\mathbf a}^L)$ are determined by the channel $P(y^L \| x^L)$.
Equations \eqref{eq:app-P} and \eqref{eq:app-H} imply that $P(y^L)$ and $H(Y^L | {\mathbf A}^L)$ are
convex combinations of $P({\mathbf a}^L)$. Furthermore, if we fix $P(y^L)$ for all $y^L$ but one,
and if we fix $H(Y^L | {\mathbf A}^L)$, then we have fixed $I({\mathbf A}^L;Y^L)$.
We can therefore focus on $\left|\set{Y}^L\right|$ constraints
and~\cite[Lemma~3.4]{Csiszar81} guarantees that we need
only $\left|\set{Y}^L\right|$ non-zero values of $P({\mathbf a}^L)$.

Similarly, observe that
\begin{align}
   P(y^L) = \sum_{x^L,\tilde{y}^L} P(x^L \| 0\tilde{y}^{L-1}) P(\tilde{y}^L,y^L \| x^L) 
\end{align}
so that if we fix $P(x^L \| 0\tilde{y}^{L-1})$ then we have fixed $P(y^L)$.
Our approach will be to replace $|\set{Y}^L|-1$ constraints of the form \eqref{eq:app-P}
with (hopefully fewer) constraints to fix $P(x^L \| 0\tilde{y}^{L-1})$.

We proceed by induction. We may fix $P(x_1)$ with $|\set{X}_1|-1$ constraints
of the form
\begin{align}
   P(x_1) = \sum_{{\mathbf a}^L} P({\mathbf a}^L) 1(x_1 | {\mathbf a}^L) .
\end{align}
This fixes $P(x_1,\tilde{y}_1)$ because the
channel specifies $P(\tilde{y}_1|x_1)$. Now suppose that
$P(x^{i-1},\tilde{y}^{i-1})$ is fixed and write
\begin{align} \label{eq:app-P2}
   P(x_i | x^{i-1},\tilde{y}^{i-1}) = \sum_{{\mathbf a}^L} P({\mathbf a}^L) \frac{P(x^i,\tilde{y}^{i-1} | {\mathbf a}^L)}{P(x^{i-1},\tilde{y}^{i-1})}
\end{align}
where $P(x^i,\tilde{y}^{i-1} | {\mathbf a}^L)$ is fixed because ${\mathbf a}^L$ is in the conditioning.
We must thus define
\begin{align} \label{eq:per-letter-bound}
   |\set{X}^{i-1}| \cdot \left| \tilde{\set{Y}}^{i-1} \right| \cdot (| \set{X}_i |-1)
\end{align}
constraints of the form \eqref{eq:app-P2} to fix $P(x_i | x^{i-1},\tilde{y}^{i-1})$ for all its arguments.
This in turn fixes $P(x_i,\tilde{y}_i | x^{i-1},\tilde{y}^{i-1})$ because the channel specifies
$P(\tilde{y}_i | x^i,\tilde{y}^{i-1})$. We thus find that $P(x^i,\tilde{y}^i)$ is fixed which completes
the induction step. Collecting all the constraints including \eqref{eq:app-H} we have
\begin{align}
   \left| \set{X}_1 \right| + \sum_{i=2}^L \left|\set{X}^{i-1}\right| \cdot \left|\set{\tilde{Y}}^{i-1}\right|
   \cdot ( \left| \set{X}_i \right| -1)
\end{align}
constraints in total. This number may be less than $|\set{Y}^L|$, e.g., if one of the $L$ channel outputs is
continuous.

\appendix{Cardinality Bounds For Relays Without Delay}
\label{app:cardinality-rc-0}
Consider an RC without delay and suppose that $P(x_1|{\mathbf a}_2)$ is specified.
This fixes $P(x_1,x_2,y_2,y_3|{\mathbf a}_2)$ because the channel fixes $P(y_2|x_1)$
and $P(y_3|x_1,x_2,y_2)$, and ${\mathbf a}_2$ specifies $1(x_2 | \mathbf{a}_2,y_2 )$
due to \eqref{eq:channel-input}. We have thus fixed $P(y_3 | {\mathbf a}_2)$,
$H(Y_3 | X_1,{\mathbf A}_2={\mathbf a}_2)$, and $I(X_1 ; Y_2 Y_3 | {\mathbf A}_2={\mathbf a}_2)$.
We further have
\begin{align}
   & P(y_3) = \sum_{{\mathbf a}_2} P({\mathbf a}_2) P(y_3 | {\mathbf a}_2) \label{eq:app-P-0} \\
   & H(Y_3 | X_1 {\mathbf A}_2) = \sum_{{\mathbf a}_2} P({\mathbf a}_2) H(Y_3 | X_1,{\mathbf A}_2={\mathbf a}_2) \label{eq:app-H-0} \\
   & I(X_1 ; Y_2 Y_3 | {\mathbf A}_2) \nonumber \\
   & \quad = \sum_{{\mathbf a}_2} P({\mathbf a}_2) I(X_1 ; Y_2 Y_3 | {\mathbf A}_2={\mathbf a}_2).\label{eq:app-I-0}
\end{align}
Finally, if we fix $P(y_3)$ for all $y_3$ but one, and if we fix $H(Y_3 | X_1 {\mathbf A}_2)$
and $I(X_1 ; Y_2 Y_3 | {\mathbf A}_2)$, then we have fixed $I(X_1 {\mathbf A}_2;Y_3)$
and (obviously) $I(X_1 ; Y_2 Y_3 | {\mathbf A}_2)$. We thus have $\left|\set{Y}_3\right|+1$ 
constraints in total to specify the bound \eqref{eq:rc-fb-0}.

Next, note that
\begin{align}
   P(y_3) = \sum_{x_1,x_2,y_2} P(x_1,x_2) P(y_2|x_1) P(y_3 | x_1,x_2,y_2) 
\end{align}
so that if we fix $P(x_1,x_2)$ then we have fixed $P(y_3)$. We proceed by writing
\begin{align}
   P(x_1,x_2) = \sum_{{\mathbf a}_2} P({\mathbf a}_2)
   P(x_1,x_2 | {\mathbf a}_2) \label{eq:app-P-02}
\end{align}
which gives us $|\set{X}_1| \cdot |\set{X}_2| -1$ constraints instead of the $|\set{Y}_3|-1$ before.
Together with \eqref{eq:app-H-0} and \eqref{eq:app-I-0} we arrive at $|\set{X}_1| \cdot |\set{X}_2|+1$
constraints in total.

\appendix{Converse for a Class of MACs with Feedback}
\label{app:mac-fb}
Let $V_i=X_1^{(i-1)L} Y^{(i-1)L}$ for $i=1,2,\ldots,m$.
Fano's inequality, $P_e \rightarrow 0$, and the independence of messages give
\begin{align}
nR_1 & \le I(W_1 ; Y^n | W_2 ) \nonumber \\
& = I({\mathbf A}_1^n ; Y^n | {\mathbf A}_2^n ) \nonumber \\
& = \sum_{i=1}^m H(Y_i^L | {\mathbf A}_2^{iL} Y^{(i-1)L}) 
    - H(Y_i^L | {\mathbf A}_1^{iL} {\mathbf A}_2^{iL} Y^{(i-1)L}) \nonumber \\
& \overset{(a)}{=} \sum_{i=1}^m H(Y_i^L | {\mathbf A}_2^{iL} V_i) 
    - H(Y_i^L | {\mathbf A}_1^{iL} {\mathbf A}_2^{iL} V_i) \nonumber \\
&\overset{(b)}{\le} m I(\bar{\mathbf A}_{1,T}^L ; Y_T^L | \bar{\mathbf A}_{2,T}^L V_T T) \nonumber \\
& \overset{(c)}{\le} m I(\bar{\mathbf A}_{1,T}^L ; Y_T^L | \bar{\mathbf A}_{2,T}^L V_T) \label{eq:conv-mac-fb-1}
\end{align}
where $(a)$ follows because ${\mathbf A}_2^i Y^{i-1}$ defines $X_2^i$ and therefore also $X_1^{i-1}$.
Step $(b)$ follows by using $T$ as our time-sharing random variable, $\bar{\mathbf A}_{k,i}^L$ as
in Appendix~\ref{app:proof}, and similar steps as in \eqref{eq:bound3};
step $(c)$ follows because
\begin{align}
  T-V_T\bar{\mathbf A}_{1,T}^L\bar{\mathbf A}_{2,T}^L-Y_T^L
\end{align}
forms a Markov chain. The chains
\begin{align}
& T-\bar{\mathbf A}_{1,T}^L\bar{\mathbf A}_{2,T}^L-Y_T^L \\
& \bar{\mathbf A}_{1,T}^L-V_T-\bar{\mathbf A}_{2,T}^L
\end{align}
are also Markov.

By symmetry, we have a similar bound as \eqref{eq:conv-mac-fb-1} for $nR_2$. The corresponding sum-rate bound is
\begin{align}
n(R_1+R_2) & \le I(W_1 W_2 ; Y^n ) \nonumber \\
& = I({\mathbf A}_1^n  {\mathbf A}_2^n ; Y^n ) \nonumber \\
& \le \sum_{i=1}^m H(Y_i^L ) 
    - H(Y_i^L | {\mathbf A}_1^{iL} {\mathbf A}_2^{iL} V_i) \nonumber \\
& = m I(\bar{\mathbf A}_{1,T}^L \bar{\mathbf A}_{2,T}^L ; Y_T^L | T) \nonumber \\
& \le m I(\bar{\mathbf A}_{1,T}^L \bar{\mathbf A}_{2,T}^L ; Y_T^L ). \label{eq:conv-mac-fb-3}
\end{align}
Collecting the bounds, we arrive at the region of Theorem~\ref{thm:mac-fb}.
The cardinality bound follows by using similar steps as in Appendices~\ref{app:cardinality-ptp}
and \ref{app:cardinality-rc-0}, see also~\cite[App.~B]{Willems:85}.

\appendix{Weakened Bound for Causal Relay Networks}
\label{app:causal-relay-networks}
The bound~\cite[Thm.~2]{Baik11}
follows from Theorem~\ref{thm:cut-set-bound} in a different way
than \eqref{eq:weak1} and \eqref{eq:weak}. We have
\begin{align}
I( {\mathbf A}_{\set{S}}^L ; Y_{\set{S}^c}^L | {\mathbf A}_{\set{S}^c}^L )
& = \sum_{i=1}^L H( Y_{\set{S}^c,i} | Y_{\set{S}^c}^{i-1} X_{\set{S}^c}^{i} {\mathbf A}_{\set{S}^c}^L )
   \nonumber \\
& \qquad \quad - H( Y_{\set{S}^c,i} | Y_{\set{S}^c}^{i-1} X_{\set{S}^c}^{i} {\mathbf A}_{\set{K}}^L )
\label{eq:weak-bk}
\end{align}
because $X_{\set{S}^c}^{i}$ is a function of $Y_{\set{S}^c}^{i-1}$ and ${\mathbf A}_{\set{S}^c}^L$. 
We bound the first entropy in the sum in \eqref{eq:weak-bk} as
\begin{align}
H( Y_{\set{S}^c,i} | Y_{\set{S}^c}^{i-1} X_{\set{S}^c}^{i} {\mathbf A}_{\set{S}^c}^L )
\le H( Y_{\set{S}^c,i} | Y_{\set{S}^c}^{i-1} X_{\set{S}^c}^{i} {\mathbf A}_{\set{S}^c \cap \set{N}_0}^L ) .
\label{eq:weak-bk-1}
\end{align}
For the second entropy in \eqref{eq:weak-bk} we use two approaches.
For $1\le i\le L-1$ we bound 
\begin{align}
H( Y_{\set{S}^c,i} | Y_{\set{S}^c}^{i-1} X_{\set{S}^c}^{i} {\mathbf A}_{\set{K}}^L )
& \ge H( Y_{\set{S}^c,i} | Y_{\set{K}}^{i-1} X_{\set{K}}^{i} {\mathbf A}_{\set{K}}^L ) \nonumber \\
& \overset{(a)}{=} H( Y_{\set{S}^c,i} | Y_{\set{K}}^{i-1} X_{\set{K}}^{i} {\mathbf A}_{\set{S}^c \cap \set{N}_0}^L )
\label{eq:weak-bk-2}
\end{align}
where $(a)$ follows because (cf.~\eqref{eq:Markov-chain-1})
\begin{align}
   {\mathbf A}_{\set{K}}^L-Y_{\set{K}}^{i-1} X_{\set{K}}^i-Y_{\set{K},i}
\end{align}
forms a Markov chain for all $i=1,2,\ldots,L$.
Next, for time $i=L$ we use
\begin{align}
& H( Y_{\set{S}^c,L} | Y_{\set{S}^c}^{L-1} X_{\set{S}^c}^{L} {\mathbf A}_{\set{K}}^L ) \nonumber \\
& \overset{(a)}{=} H( Y_{\set{S}^c,L} | Y_{\set{S}^c}^{L-1} X_{\set{S}^c}^{L} {\bf A}_{\set{S}}^L {\mathbf A}_{\set{S}^c \cap \set{N}_0}^L) \nonumber \\
& = H( Y_{\set{S}^c,L} | Y_{\set{S}^c}^{L-1} X_{\set{S}^c}^{L}
{\bf A}_{\set{S} \cap \set{N}_1}^L {\mathbf A}_{\set{S} \cap \set{N}_0}^L {\mathbf A}_{\set{S}^c \cap \set{N}_0}^L ) \nonumber \\
& \overset{(b)}{=} H( Y_{\set{S}^c,L} | Y_{\set{S}^c}^{L-1} X_{\set{S}^c}^{L}
X_{\set{S} \cap \set{N}_1}^L {\mathbf A}_{\set{S} \cap \set{N}_0}^L {\mathbf A}_{\set{S}^c \cap \set{N}_0}^L )
\label{eq:weak-bk-3}
\end{align}
where $(a)$ follows because
\begin{align}
  {\mathbf A}_{\set{S}^c}^L-Y_{\set{S}^c}^{L-1} X_{\set{S}^c}^L {\mathbf A}_{\set{S}}^L - Y_{\set{S}^c,L}
\end{align}
forms a Markov chain, and $(b)$ follows because $X_{\set{S} \cap \set{N}_1}^L$ is a function of ${\mathbf A}_{\set{S} \cap \set{N}_1}^L$ and because
\begin{align}
  {\mathbf A}_{\set{S} \cap \set{N}_1}^L-Y_{\set{S}^c}^{L-1} X_{\set{S}^c}^L X_{\set{S} \cap \set{N}_1}^L {\mathbf A}_{\set{S} \cap \set{N}_0}^L {\mathbf A}_{\set{S}^c \cap \set{N}_0}^L - Y_{\set{S}^c,L}
\end{align}
forms a Markov chain.

Summarizing, we insert \eqref{eq:weak-bk-1}, \eqref{eq:weak-bk-2}, and \eqref{eq:weak-bk-3} into \eqref{eq:weak-bk}
and obtain the following bound that appeared in~\cite[Thm.~2]{Baik11}:
\begin{align}
& I( {\mathbf A}_{\set{S}}^L ; Y_{\set{S}^c}^L | {\mathbf A}_{\set{S}^c}^L) \nonumber \\
& \le I( X_{\set{S}}^{L-1}, 0Y_{\set{S}}^{L-2} \rightarrow Y_{\set{S}^c}^{L-1} \| X_{\set{S}^c}^{L-1} | {\mathbf A}_{\set{S}^c \cap \set{N}_0}^L ) \nonumber \\
& \quad + I( X_{\set{S} \cap \set{N}_1}^L {\mathbf A}_{\set{S} \cap \set{N}_0}^L ; Y_{\set{S}^c,L} | Y_{\set{S}^c}^{L-1} X_{\set{S}^c}^{L} {\mathbf A}_{\set{S}^c \cap \set{N}_0}^L ).
\label{eq:weak-bk-4}
\end{align}

\section*{Acknowledgments}
The results reported here were motivated by the paper~\cite{FongIT:12} by S.\ Fong and R.\ Yeung.
I am grateful to these authors for sending me of an early version of their work.
I am also grateful to S.-Y.\ Chung, H.\ Permuter, Y.-H.\ Kim, and the two reviewers
for their detailed and constructive criticisms of the document.

\bibliographystyle{IEEE} 



\end{document}